\documentclass[sigconf]{aamas}

\usepackage{balance} 

\usepackage{graphicx}
\usepackage{booktabs} 
\usepackage{acro}
\usepackage{amsmath}
\usepackage{amsthm}
\usepackage{siunitx}
\usepackage{bm}
\usepackage{bbm}
\usepackage{tcolorbox}
\usepackage{nicefrac, xfrac}
\usepackage{mathtools}
\usepackage{subcaption}
\usepackage{dsfont}
\usepackage{tikz}
\usetikzlibrary{trees}
\usetikzlibrary{quotes}
\usetikzlibrary{positioning}
\usepackage{oplotsymbl}
\usepackage{marvosym}
\usepackage{savesym}
\savesymbol{Cross}
\usepackage{bbding}
\usepackage{multicol,multirow}
\usepackage{rotating}
\restoresymbol{bb}{Cross}
\usetikzlibrary{trees}
\usetikzlibrary{quotes}
\usetikzlibrary{shapes}
\usepackage{xspace}
\usepackage{color}
\usepackage{enumitem}
\usepackage{cleveref}
\crefrangeformat{appendix}{#3#1#4--#5#2#6}

\usepackage{algorithmicx}
\usepackage[Algorithm,ruled,section]{algorithm}
\usepackage{algpseudocode}
\usepackage{etoolbox}
\usepackage{url}
\usepackage{hyperref}

\usepackage{xr}

\newcommand{\term}[1]{\textbf{\textit{#1}}}
\newcommand{\T}{\textsc{true}}
\newcommand{\F}{\textsc{false}}

\newcommand{\cI}{\mathcal{I}}

\newcommand{\cV}{\mathcal{V}}

\newcommand{\cR}{\mathcal{R}}
\newcommand{\cS}{\mathcal{S}}
\newcommand{\bsigma}{\bm{\sigma}}

\newcommand{\bbR}{\mathbb{R}}
\newcommand{\regret}{\mathsf{Reg}}
\newcommand{\order}{\mathcal{P}}
\newcommand{\dondG}{$\textsc{Bargain}$\xspace}
\newcommand{\numItems}{m}
\newcommand{\thresh}{\nu}
\newcommand{\offer}{\omega}

\newcommand{\gengoof}{$\textsc{GenGoof}$\xspace}
\newcommand{\gengoofK}[1]{$\textsc{GenGoof}_{#1}$\xspace}
\newcommand{\modgg}{$\textsc{PrivateGenGoof}$\xspace}
\newcommand{\modggK}[1]{$\textsc{PrivateGenGoof}_{#1} $\xspace}
\newcommand{\given}{\hspace{0.1em} \vert \hspace{0.1em}}

\newcommand{\Udist}{\mathbb{U}}

\DeclareMathOperator*{\argmax}{arg\,max}

\newcommand{\Prob}{\mathsf{Pr}}

\definecolor{player1}{RGB}{213, 94, 0}
\definecolor{player1_1}{RGB}{100, 143, 255}
\definecolor{player1_2}{RGB}{0, 153, 255}
\definecolor{player1_3}{RGB}{226, 146, 57}
\definecolor{player1_4}{RGB}{161, 191, 1}
\definecolor{player1_5}{RGB}{121, 71, 186}
\definecolor{player1_6}{RGB}{193, 32, 224}
\definecolor{player1_7}{RGB}{122, 4, 4}
\definecolor{player1_8}{RGB}{37, 124, 163}
\definecolor{player1_9}{RGB}{210, 153, 94}
\definecolor{player1_10}{RGB}{94, 210, 153}
\definecolor{player1_11}{RGB}{107, 65, 147}
\definecolor{player1_12}{RGB}{184, 126, 122}
\definecolor{player1_13}{RGB}{176, 203, 166}
\definecolor{player1_14}{RGB}{143, 175, 107}
\definecolor{player1_15}{RGB}{232, 244, 218}
\definecolor{player1_16}{RGB}{145, 76, 61}
\definecolor{player1_17}{RGB}{193, 94, 194}
\definecolor{player1_18}{RGB}{153, 0, 0}
\definecolor{player1_19}{RGB}{0, 153, 153}
\definecolor{player1_20}{RGB}{48, 233, 186}

\definecolor{player2}{RGB}{0, 114, 189}
\definecolor{player2_1}{RGB}{135, 100, 255}
\definecolor{player2_2}{RGB}{176,224,230}
\definecolor{player2_3}{RGB}{128, 0, 128}
\definecolor{player2_4}{RGB}{255, 178, 255}
\definecolor{player2_5}{RGB}{205,188,198}
\definecolor{player2_6}{RGB}{214, 198, 164}
\definecolor{player2_7}{RGB}{204, 0, 204}
\definecolor{player2_8}{RGB}{204, 204, 0}
\definecolor{player2_9}{RGB}{184, 173, 241}
\definecolor{player2_10}{RGB}{64, 224, 208}
\definecolor{player2_11}{RGB}{204, 255, 0}
\definecolor{player2_12}{RGB}{229, 193, 0}
\definecolor{chance}{RGB}{0, 158, 115}
\definecolor{otherplayer1}{RGB}{204, 121, 167}
\definecolor{otherplayer2}{RGB}{240, 228, 66}
\definecolor{terminal}{RGB}{253, 236, 160}
\definecolor{oldplayer}{RGB}{150, 150, 150}

\definecolor{gg_player1}{RGB}{128, 0, 128}
\definecolor{gg_player1_1}{RGB}{210, 151, 139}
\definecolor{gg_player1_2}{RGB}{192, 222, 150}
\definecolor{gg_player1_3}{RGB}{0, 105, 150}
\definecolor{gg_player1_4}{RGB}{255, 105, 180}
\definecolor{gg_player1_5}{RGB}{142, 59, 0}
\definecolor{gg_player1_6}{RGB}{152, 152, 255}
\definecolor{gg_player1_7}{RGB}{243, 222, 138}
\definecolor{gg_player1_8}{RGB}{1, 137, 223}
\definecolor{gg_player1_9}{RGB}{214, 85, 85}
\definecolor{gg_player1_10}{RGB}{218, 165, 32}
\definecolor{gg_player1_11}{RGB}{80, 200, 120}
\definecolor{gg_player1_12}{RGB}{120, 80, 200}
\definecolor{gg_player1_13}{RGB}{204,161,201}
\definecolor{gg_player1_14}{RGB}{190,209,227}
\definecolor{gg_player1_15}{RGB}{238,251,29}

\definecolor{gg_player2}{RGB}{230, 0, 230}
\definecolor{gg_player2_1}{RGB}{65, 105, 225}
\definecolor{gg_player2_2}{RGB}{68, 170, 153}
\definecolor{gg_player2_3}{RGB}{230, 97, 0}

\definecolor{gg_chance}{RGB}{126, 127, 154}

\setcopyright{ifaamas}
\acmConference[AAMAS '26]{Proc.\@ of the 25th International Conference
on Autonomous Agents and Multiagent Systems (AAMAS 2026)}{May 25 -- 29, 2026}
{Paphos, Cyprus}{C.~Amato, L.~Dennis, V.~Mascardi, J.~Thangarajah (eds.)}
\copyrightyear{2026}
\acmYear{2026}
\acmDOI{}
\acmPrice{}
\acmISBN{}

\acmSubmissionID{937}

\title[PBE computation with application to EGTA]{Computing Perfect Bayesian Equilibria, with Application to Empirical Game-Theoretic Analysis}

\author{Christine Konicki$^*$}\thanks{$^*$Konicki worked on this paper while a PhD student at the University of Michigan}
\affiliation{
  \institution{Michigan Tech Research Institute}
  \city{Ann Arbor}
  \country{USA}}
\email{ckonicki@mtu.edu}

\author{Mithun Chakraborty}
\affiliation{
  \institution{University of Michigan}
  \city{Ann Arbor}
  \country{USA}}
\email{dcsmc@umich.edu}

\author{Michael P. Wellman}
\affiliation{
  \institution{University of Michigan}
  \city{Ann Arbor}
  \country{USA}}
\email{wellman@umich.edu}

\begin{abstract}
Perfect Bayesian Equilibrium (PBE) is a refinement of the Nash equilibrium for imperfect-information extensive-form games (EFGs) that enforces consistency between the two components of a solution: agents’ strategy profile describing their decisions at information sets and the belief system quantifying their uncertainty over histories within an information set. We present a scalable approach for computing a PBE of an arbitrary two-player EFG. We adopt the definition of PBE enunciated by Bonanno in 2011 using a consistency concept based on the theory of belief revision due to Alchourr\'{o}n, G\"{a}rdenfors, and Makinson. Our algorithm for finding a PBE is an adaptation of Counterfactual Regret Minimization (CFR) that minimizes the expected regret at each information set given a belief system, while maintaining the necessary consistency criteria. We prove that our algorithm is correct for two-player zero-sum games and has a reasonable slowdown in time-complexity relative to classical CFR given the additional computation needed for refinement. We also experimentally demonstrate the competent performance of PBE-CFR in terms of equilibrium quality and running time on medium-to-large non-zero-sum EFGs. Finally, we investigate the effectiveness of using PBE for strategy exploration in empirical game-theoretic analysis. Specifically, we compute PBE as a meta-strategy solver (MSS) in a tree-exploiting variant of Policy Space Response Oracles (TE-PSRO). Our experiments show that PBE as an MSS leads to higher-quality empirical EFG models with complex imperfect information structures compared to MSSs based on an unrefined Nash equilibrium.
\end{abstract}


\newcommand{\BibTeX}{\rm B\kern-.05em{\sc i\kern-.025em b}\kern-.08em\TeX}

\begin{document}


\pagestyle{fancy}
\fancyhead{}


\maketitle 


\section{Introduction}
Game theory offers a variety of approaches for formally representing strategic interactions among several autonomous agents and reasoning about their outcomes with the help of \textit{solution concepts}. The preeminent game-theoretic solution concept is the Nash Equilibrium (NE), a strategy profile such that no agent can improve its payoff by unilaterally deviating from the profile. Ever since its introduction and the proof of its guaranteed existence for finite games \cite{Nash51}, the NE has been the focus of several threads of theoretical and empirical research.

An important thread concerns \textit{refinements} of the NE for \textit{extensive-form games} (EFGs), tree-based representations of dynamic multiagent interactions that explicitly capture the sequential nature of action-taking and conditioning on observations. In general, the NE of a game is non-unique, and a refinement is a set of criteria that selects plausible outcomes from among all Nash equilibria, given the characteristics of a class of games. The \textit{subgame perfect equilibrium} (SPE) \cite{selten65} is a natural refinement
for an EFG with perfect information (i.e., when every agent knows the full game history leading up to each of its decision points). An SPE rules out non-credible threats by requiring the solution to induce a NE in each subgame. 
Under imperfect information, EFGs use the device of \textit{information sets} to represent agents' inability to distinguish certain game histories; 
\citet{kaminski19} generalized SPEs to (potentially infinite) imperfect-information EFGs by refining the definition of subgames such that every information set is contained within a single subgame.

The most powerful NE refinements for imperfect-information EFGs augment the game solution space from that of strategy profiles to that of \textit{assessments}. An assessment consists of a strategy profile and a \textit{belief system}, a quantification of each agent's uncertainty over all decision points in each of its information sets via probability distributions. To be an equilibrium, an assessment must meet two conditions. First, it must satisfy \textit{sequential rationality}, which stipulates that no unilateral deviation can improve expected utility at any information set. Second, all distributions induced by the assessment's strategies and beliefs must conform to Bayes' rule. However, game theorists have also given much thought to additional notions of \textit{consistency} to be enforced \textit{between} the two components of an assessment to address further plausibility issues, resulting in a few different refined solution concepts. \citet{kw82} proposed the \textit{sequential equilibrium} (SE) that satisfies a topological consistency notion, called KW-consistency by \citet{agm11}, which is unintuitive and hard to verify. A simplification of SE called the \textit{weak sequential equilibrium} \cite{myerson91} imposes conformity to Bayes' rule only on information sets reached with positive probability under the strategy profile.  \citet{pbe91} introduced a solution concept of intermediate strength that they termed \textit{perfect Bayesian equilibrium}, but they demonstrated its construction only for a restricted class of games called multi-stage signaling games. A major issue with the practical implementation of these weaker equilibria is that the lack of concrete, general consistency restrictions on off-equilibrium paths makes them unsuitable as bases for general-purpose game-solving algorithms.  

Bonanno \cite{bonanno_plaus11,agm11} introduced a new consistency notion that covers both on- and off-equilibrium paths, based on the theory of belief revision due to Alchourr\'{o}n, G\"{a}rdenfors, and Makinson \cite{AGM85}, which he termed \textit{AGM-consistency} after the original authors. 
This notion requires the concept of \textit{plausibility orders} over the nodes (representing histories or, equivalently, decision points) of the EFG tree based entirely on structural properties (edge incidence and information set membership). An assessment is AGM-consistent if it is possible to construct a plausibility order such that positive probabilities are assigned to nodes or edges of the EFG tree by the assessment if and only if they satisfy certain relationships under the order in question. \citet{agm11} reused the same term `perfect Bayesian equilibrium' (PBE) for this refinement using AGM-consistency, which is still weaker than KW-consistency but easier to algorithmically verify in principle---we will use this definition of PBE in this paper. Although it is known that every finite EFG admits at least one PBE \cite{agm11}, the implementation and evaluation of robust, scalable algorithmic approaches towards the computation of a PBE for arbitrary dynamic games of imperfect information is an important practical question. This is the primary research question that motivates this work.

The framework of empirical game-theoretic analysis (EGTA) \citep{wellman2025empirical} is highly relevant to our work. 
For multiagent scenarios that are too complicated for an analytic representation but admit procedural descriptions that can be queried (e.g., a simulator), EGTA offers a toolkit for using data collected from such queries to estimate a coarser model, called an \textit{empirical game}; a key idea is to make this model amenable to off-the-shelf game-solving algorithms so that approximate insights about the underlying scenario can be derived from it. 
A popular and powerful iterative approach to EGTA called \textit{policy space response oracles} (PSRO) \cite{bighashdel_psro24,psro17} uses an arbitrary game-solving algorithm as a module called the \textit{meta-strategy solver} (MSS), which provides a principled basis for exploring the underlying strategy space to augment the model. 
Whereas the empirical game in EGTA has commonly been maintained in the less expressive normal form, we developed a \textit{tree-exploiting} variant of EGTA (TE-EGTA) in prior work \citep{konicki22,konicki25} that represents the empirical game as an EFG tree. 
This approach makes use of refined solution concepts feasible for empirical games and potentially conducive to higher-quality game models, and the MSS for a tree-based model provides a natural use case for our algorithmic contributions in this paper. 



\subsection{Our Contributions}

We propose a novel practical algorithm PBE-CFR (Algorithms~\ref{alg:pbe_cfr}, \ref{alg:pbe_traverse}, \ref{alg:update_beliefs}) for computing a PBE \cite{agm11} in arbitrary two-player EFGs. It is a non-trivial adaptation of the classic Counterfactual Regret Minimization (CFR) algorithm \cite{cfrm07} that minimizes the expected regret at each information set given a belief system, while enforcing AGM-consistency.
\begin{itemize}[leftmargin=*]
    \item For two-player zero-sum games, we prove that the algorithm is correct by establishing a guarantee of convergence to an exactly sequentially rational solution, and analyze its space and time complexity (Section~\ref{sec:pbe_proof}).
    \item For two qualitatively different classes of two-player general-sum games, we experimentally demonstrate that PBE-CFR performs competently in practice both in approximation quality and running time (Sections~\ref{sec:setup} and \ref{sec:pbe_cfr_time_exp}).
\end{itemize}
  We also report experiments that demonstrate the usefulness of PBE-CFR, vis-\`a-vis an unrefined NE obtained by classical CFR, for strategy exploration in TE-PSRO \citep{konicki22,konicki25}.
  In particular, we characterize how the speed of convergence to zero of the regret of the TE-PSRO empirical model with PBE as the MSS depends on the degree of coarsening of the information structure of the underlying game (Section~\ref{sec:pbe_exp_tepsro}). The code for our implementation of PBE-CFR and all our experiments can be found at \url{https://github.com/ckonicki-umich/AAMAS26/}.

\subsection{Further Related Work}\label{sec:relwork}
We provide a more detailed review of consistency concepts for NE refinements for EFGs in Apps.~\crefrange{sec:seq_eq}{sec:pbe_msg}. \citet{wellman2025empirical} provides an overview of EGTA techniques including PSRO. A body of work exists on  algorithms for computing (approximate) NE refinements \citep{azhar05,turocy2010computing,panozzo14,Thoma25,graf_sequent_eq24}, but the scalability and practicality of these algorithms has not been adequately established, to the best of our knowledge; please see App.~\ref{sec:related_chap7} for further details. There is a rich literature on extensions of the CFR approach including warm-start CFR \cite{brown_cfr14}, CFR$^+$ \cite{cfrplus14}, CFR-D \cite{zun_subgame4}, discounted CFR \cite{brown2019superhuman}, linear CFR \cite{brown2019solving}, deep CFR \cite{dcfrm_19}, Monte Carlo CFR \cite{mccfr09}, PCFR$^+$ \cite{farina2021faster}, and dynamic discounted CFR \cite{xu2024dynamic}.
In prior work, we developed a scalable, modular implementation \cite{konicki25} of the generalized backward induction algorithm \cite{kaminski19} for computing an SPE of an imperfect-information EFG, and found in experiments that TE-PSRO with an SPE as MSS converges to a high-quality model faster than with an unrefined NE as MSS for diverse game classes. A similar treatment of PBE is a natural next step.



\section{Technical Preliminaries}\label{sec:prelim}
A finite imperfect-information \term{extensive-form game} (EFG) is a tuple $G:=\langle N$, $H$, $V$, $\{\mathcal{I}_j\}_{j=0}^n$, $\{A_j\}_{j=1}^n$, $X$, $P$, $u \rangle$, where
\begin{itemize}[leftmargin=*]
    \item $N = \{0, \dotsc, n\}$ is the player set. Player~$0$ denotes \term{Nature}, a non-strategic agent responsible for chance events that impact the course of play.
    \item $H$, the \term{game tree}, is a finite tree rooted at node $h_0$ that captures players' dynamic interactions. Each node $h \in H$ represents a \term{state} or, equivalently, a \term{history} of the game beginning at $h_0$ (which corresponds to the empty history $\emptyset$). The \term{terminal nodes} $Z \subset H$ or leaves of the game tree represent possible end-states of the game. The remaining nodes $D = H \setminus Z$ are \term{decision nodes}. 
    \item $V: D \rightarrow N$ assigns a player to each decision node $h$. A node $h$ where $V(h) = 0$ is called a \term{chance node}.
    \item For each player~$j \in N$, the set $\cI_j$ is a partition of $V^{-1}(j)$ where each $I \in \cI_j$ is an \term{information set} of $j$. All nodes $h \in I$ are indistinguishable to player~$j$. $I(h)$ denotes the information set to which a node $h$ belongs. We assume all information sets to be consistent with perfect recall \citep[Definition 5.2.3]{shoham2008multiagent}. 
    \item $A_j(I)$ denotes actions that player~$j$ can take at information set $I \in \cI_j$.
    \item $X(h)$ is the set of possible outcomes of Nature's stochastic event at $h$.
    \item $P(\cdot \given h)$ is the probability distribution over $X(h)$.
    \item  $u: Z \rightarrow \bbR^n$ maps each terminal node $z$ to a vector of players' \term{utilities} $\{u_j(z)\}_{j=1}^n$.
\end{itemize}
The directed edge connecting any $h \in I$ to its child represents a state transition resulting from $V(h)$'s move and is labeled with an outcome $x \in X(h)$ if $V(h) = 0$ or an action $a \in A_{V(h)}(I)$ otherwise, the child-node being denoted by $hx$ or $hs$ respectively. The set of nodes within $H$ that succeed a given node $h$ is denoted by $\mathsf{Succ}(h)$. The function $\varphi$ maps each node $h \in I \in \cI_j$ to the ordered sequence of actions and chance outcomes observable to $j$ from the root node leading up to $I$, according to the designated rules of the (imperfect-information) game. When the input is a terminal node $z \in Z$, which does not belong to any information set, $\varphi$ returns a complete history from $z$ to the root node, or the sequence of a specific player's actions given $z$ and $j \in N$.

A \term{pure strategy} $\pi_j(\cdot)$ for player~$j \in N \setminus \{0\}$ specifies the action $a \in A_j(I)$ that $j$ selects at each information set $I \in \cI_j$.
More generally, a \term{mixed strategy} $\sigma_j(\cdot \given I)$ defines a probability distribution over $A_j(I)$ at each information set of agent $j$ where an action $a \in A_j(I)$ is selected with probability $\sigma_j(a \given I)$.
A \term{strategy profile} is a vector $\bsigma = (\sigma_1, \dotsc, \sigma_n)$, and $\bsigma_{-j}$ denotes the collection of strategies of all players other than $j$ in $\bsigma$. 
$\Sigma_j$ denotes the set of all strategies available to player~$j$, and $\Sigma = \times_{j = 1}^n \Sigma_j$ the space of  strategy profiles.

The likelihood that node $h \in H$ is reached by strategy profile $\bsigma$ is given by its \term{reach probability}
\begin{equation*}
r(h, \bsigma) \coloneq r_0(h) \prod_{j \in N \setminus \{0\}} r_j(h, \sigma_j),
\end{equation*}
where $r_j(h, \sigma_j)$ is the joint probability of player~$j$ choosing actions that lead to $h$ according to $\sigma_j$ at each of its decision nodes on the path to $h$; 
Nature's contribution $r_0(h)$ is the joint probability of each chance node along the path to $h$ producing an outcome 
leading to $h$. The reach probability of information set $I$ under $\bsigma$ is $r(I, \bsigma) = \sum_{h\in I} r(h, \bsigma)$. 
A node or information set with a positive reach probability is said to be \textbf{\textit{reachable}} under the given strategy profile. 
The \term{expected utility} of player~$j$ 
under a strategy profile $\bsigma$ is given by 
\begin{equation*}
U^E_j(\bsigma) \coloneq \sum_{z \in Z} u_j(z) r(z, \bsigma). 
\end{equation*}
The \term{regret} of player~$j$ at profile $\bsigma$ is given by 
\[\regret_j(\bsigma) = \max_{\sigma \in \Sigma_j}~U^E_j(\sigma, \bsigma_{-j}) - U^E_j(\bsigma).\] 
We define the regret of a profile as the sum of player regrets, that is $\regret(\bsigma) = \sum_{j=1}^n \regret_j(\bsigma)$.
A strategy profile $\bsigma$ with $\regret(\bsigma) = 0$ is a \term{Nash equilibrium} (NE).

For any $h$ that precedes a terminal node $z$, we denote by $r(z \given h, \bsigma)$ the conditional reach probability of $z$ according to $\bsigma$, given that $h$ has already been reached. That is, $r(z \given h, \bsigma) = r(z, \bsigma)/r(h, \bsigma)$ whenever $h$ is reachable under $\bsigma$ and $z\in\mathsf{Succ}(h)$ as well as the joint probability of all players choosing the right actions that lead to $z$ starting from state $h$ according to $\bsigma$.
Moreover, the \term{conditional expected utility} of player $j$ given that it is at a node $h$ under a strategy profile $\bsigma$ is given by  
\begin{equation*}
U^E_j(\bsigma \given h) \coloneq \sum_{z \in Z} u_j(z) r(z \given h, \bsigma) = \sum_{z \in Z \cap \mathsf{Succ}(h)} u_j(z) r(z \given h, \bsigma). 
\end{equation*}

For an EFG with 
at least one non-singleton information set, players' uncertainty about game states is naturally captured by a \term{system of beliefs} denoted by $\mu$ and defined as a collection of probability distributions, one for each information set $I$. 
At an information set $I \in \cI_j$ of player $j \in N \setminus \{0\}$, $\mu(\cdot \given I)$ represents player $j$'s belief about which tree node it is actually at; $\mu(h\given I) \ge 0$ for every $h\in I$, and $\sum_{h \in I}\mu(h\given I)=1$. 
An ordered pair $(\bsigma, \mu)$ containing a strategy profile $\bsigma$ and a system of beliefs $\mu$ is called an \term{assessment} and serves as a solution candidate for an imperfect-information EFG. App.~\ref{app:main_example} provides an example illustrating EFGs and assessments.  

\subsection{Perfect Bayesian Equilibrium}\label{sec:pbe}
We now present the definition of the \term{perfect Bayesian equilibrium} (PBE) 
proposed by \citet{agm11}. 
The three defining properties of a PBE are sequential rationality, AGM-consistency, and compatibility of beliefs with Bayes' rule throughout the game tree.


Sequential rationality is the natural extension of subgame perfection from strategies to assessments. It stipulates that an assessment must induce an NE at each player's information set, conditioned on both the player's belief distribution at that information set and the assumption that the information set has been reached during gameplay. Let $U^B_j\left(\bsigma, \mu \given I\right)$ denote the \term{believed utility} of player $j$ at information set $I \in \mathcal{I}_j$ for playing strategy $\sigma_j$ while the others play the profile $\bsigma_{-j}$, 
given its belief $\mu(\cdot \given I)$; i.e.,
\begin{align}
    U^B_j\left(\bsigma, \mu \given I\right) &\coloneq \sum_{h \in I} \sum_{z \in Z} \mu(h \given I)\ r(z \given h, \bsigma)\ u_j(z)\notag\\
    &= \sum_{h \in I} \mu(h \given I)\ U^E_j\left( \bsigma \given h \right)\notag \\
    &= \sum_{a \in A_j(I)} \sigma_j(a \given I) \left( \sum_{h \in I} \mu(h \given I)\ U^E_j\left( \bsigma \given ha \right) \right) \label{UB_comp} 
\end{align}


\begin{definition}[Sequential Rationality]\label{def:seq_rat}
An assessment $(\bsigma, \mu)$ is \term{sequentially rational} if, at every information set $I \in \cI_j$ of each player $j \in N \setminus \{0\}$, 
\begin{equation*}
    U^B_j\left(\bsigma, \mu \given I\right) \geq U^B_j\left(\sigma'_j, \bsigma_{-j}, \mu \given I\right),\ \forall \sigma'_j \in \Sigma_j.
\end{equation*}
It is sufficient to restrict $\sigma'_j$ to pure strategy deviations at $I$ \citep{morrill21_seqrat}. 
\end{definition}



The AGM-consistency criterion is based on the concept of a \term{plausibility order} over the nodes in $H$ defined as follows \citep{agm11}.
\begin{definition}[Plausibility Order]\label{def:plaus}
  A plausibility order is a total preorder $\precsim$ on the set $H$ that satisfies the following conditions:\footnote{We say that node $a$ is \textit{at least as plausible} as node $b$ if $a \precsim b$; the symbols $\prec$ and $\sim$ have standard meanings given preorder $\precsim$.}
\begin{itemize}[leftmargin=*]
    \item For any node $h \in D$ and any action $a \in A_{V(h)}(I(h))$, it is impossible that $ha \prec h$.
    \item Every node $h \in D$ has at least one action $a \in A_{V(h)}(I(h))$ such that $ha \precsim h$; each $a$ that satisfies $ha \precsim h$ also satisfies $h'a \precsim h'$ for all $h' \in I(h)$.
    \item For every chance node $h$ and every outcome $e \in X(h)$, $he \precsim h$.
\end{itemize} 
\end{definition}
 Given a history $h$, we say that plausibility is \textit{preserved} in another history $h' \in \mathsf{Succ}(h)$ if $h' \precsim h$. 


\begin{definition}[AGM-consistency \citep{agm11}]\label{def:agm}
An assessment $(\bsigma, \mu)$ for game $G$ is \term{AGM-consistent} if a plausibility order $\order$ can be constructed on $H$ such that:
\begin{itemize}[leftmargin=*]
    \item For each node $h \in H$ and action $a \in A_{V(h)}(I(h))$, $\bsigma(a) > 0$ if and only if $h \sim ha$ in $\order$; 
    \item For each chance node $h \in H$ and possible chance outcome $x \in X(h)$, $P(x \given h) > 0 $ if and only if $ h \sim hx$ in $\order$;
    \item For each node $h \in H$, $\mu(h \given I(h)) > 0 $ if and only if $ h \precsim h'$ in $\order$ for all $h' \in I(h)$.
\end{itemize} 
\end{definition}
A plausibility order $\order$ satisfying the three conditions in Definition~\ref{def:agm} is said to \term{rationalize} the assessment $(\bsigma, \mu)$.

\begin{figure}[ht!]
    \centering
    \includegraphics[scale=0.27]{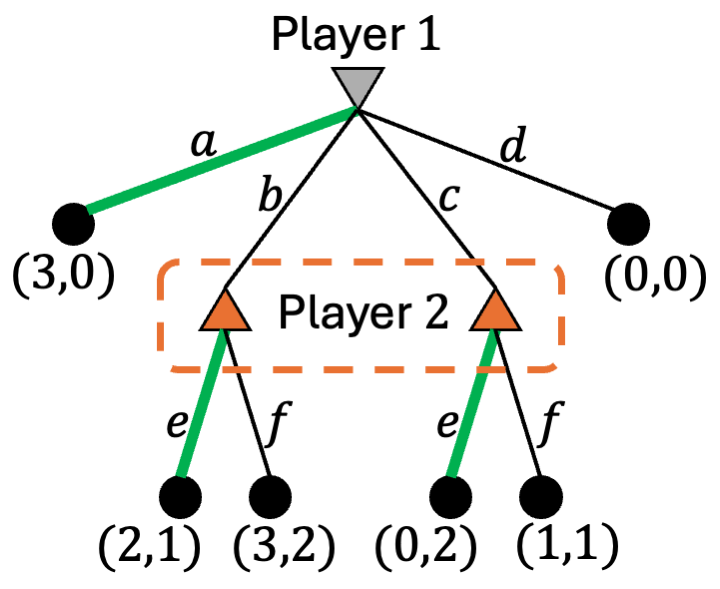}
    \caption{Example of an imperfect-information EFG from \citet{agm11}, augmented with leaf utilities. There is one non-singleton information set (for Player~$2$) represented by the orange box. The equilibrium path induced by the AGM-consistent assessment $(\bm{\sigma}^*, \mu^*)$ described in Example~\ref{ex:agm_consistency}  is highlighted in green.}
    \label{fig:agm_consist}
\end{figure}

With this background in place and assuming familiarity with Bayes' rule, we now furnish the definition of PBE that we will use in the rest of the paper.

\begin{definition}[Perfect Bayesian Equilibrium \citep{agm11}]\label{def:pbe_agm11}
    An assessment $(\bsigma, \mu)$ is a \term{perfect Bayesian equilibrium} for a given imperfect-information game $G$ if it satisfies sequential rationality (Definition~\ref{def:seq_rat}) and AGM-consistency (Definition~\ref{def:agm}), and every distribution in $\mu$ follows Bayes' rule given $\bsigma$; specifically, at every reachable information set $I \in \cI_j$ for every player $j \in N \setminus \{ 0 \}$ and every $h \in I$, $\mu(h) = \frac{r(h, \bsigma)}{r(I, \bsigma)} = \frac{r(h, \bsigma)}{\sum\limits_{h' \in I} r(h', \bsigma)}$.
\end{definition}

The example below illustrates a PBE of an imperfect-information EFG with emphasis on its AGM-consistency; for an example of an assessment violating AGM-consistency, see App.~\ref{app:agm_example}.

\begin{example}\label{ex:agm_consistency}
Consider the $2$-player, imperfect-information EFG depicted in Figure~\ref{fig:agm_consist}. Let $(\bm{\sigma}^*, \mu^*)$ be an assessment of this EFG where $\bsigma^*$ assigns a probability of $1$ to each of actions $a$ and $e$ and $0$ to every other edge, and $\mu^*(c) = 1$.
To rationalize $(\bm{\sigma}^*, \mu^*)$, a plausibility order must require that $a \sim \emptyset$ and $a \precsim b, c, d$, since $\sigma^*_1(a) = 1$. Likewise, since $\sigma^*_2(e) = 1$, $b \sim \mathit{be}$ and $c \sim \mathit{ce}$. By extension, this means that $b \precsim \mathit{bf}$ and $c \precsim \mathit{cf}$. Since $\mu^*(c) = 1$ (and hence $\mu^*(b) = 0$), we must have that $c \precsim b$. Moreover, transitivity entails that $\mathit{be} \precsim \mathit{bf}$ and $\mathit{ce} \precsim \mathit{cf}$. No contradictions arise in this construction; in fact, there are multiple plausibility orders that rationalize $(\bm{\sigma}^*, \mu^*)$  
depending on where nodes $\mathit{cf}$ and $d$ are placed in the order. Therefore, $(\bm{\sigma}^*, \mu^*)$ satisfies AGM-consistency; it trivially conforms to Bayes' rule, and it can be checked algebraically from definitions that it is also sequentially rational.
\end{example}


\section{Algorithm for Finding PBE}\label{sec:alg_pbe}

Before presenting our main algorithmic contribution, we will mention a collection of algorithms that we devised to verify whether a given assessment is a PBE of a given imperfect-information EFG, each focusing on one of the three conditions in Definition~\ref{def:pbe_agm11}. We present pseudocode and written descriptions of verification methods \textsc{IsSequentRational}, \textsc{SatisfiesBayes}, and \textsc{IsConsistent}, respectively, in App.~\ref{app:algo}; since PBE-CFR applies to two-player EFGs, we present the two-player versions of these procedures as Algorithms~\ref{alg:verify_seq_rat},~\ref{alg:verify_bayes}, and~\ref{alg:verify_agm} respectively, but they can be naturally extended to an arbitrary number of players. 
In the rest of the paper, we will sometimes use $\bsigma(I)(a)$ to denote the probability assigned to action $a$ at information set $I$ by the strategy profile $\bsigma$ (i.e., $\sigma_j(a|I)$ where $j$ is the player active at $I$). 

We now present our central contribution PBE-CFR, an algorithm for computing a PBE of a given EFG; Algorithms~\ref{alg:pbe_cfr},~\ref{alg:pbe_traverse}, and~\ref{alg:update_beliefs} provide the pseudocode for the main algorithm and its subroutines. PBE-CFR is an adaptation of CFR that minimizes what we call the \term{believed regret} of playing $\bsigma$ at each information set given a belief system $\mu$ while keeping $\mu$ consistent. 

Let $U^B_j \left( \mu^t, \left.\bsigma^t\right|_{I \rightarrow a} \given I \right)$ denote the \term{believed action utility} of playing action $a$ at $I$ in iteration $t$ of the algorithm. It can be computed in a way similar to $U^B_j\left(\bsigma, \mu \given I\right)$ in Equation~\eqref{UB_comp} except for marginalization over $A_j(I)$. In addition, we define
\[R^T_{j, \mathit{imm}}(I)(a) \coloneq \frac{1}{T} \sum_{t=1}^T \left[ U^B_j \left( \mu^t, \left.\bsigma^t\right|_{I \rightarrow a} \given I \right) - U^B_j \left( \mu^t, \bsigma^t \given I \right) \right]\]
Then, the \term{immediate believed regret} of playing $\bsigma$ at information set $I$ at timestep $T$ is  given by
\begin{align*}
    R^T_{j, \mathit{imm}}(I) &\coloneq \max\limits_{a \in A_j(I)} R^T_{j, \mathit{imm}}(I)(a)
 \end{align*}
 An action $a^*\in\argmax R^T_{j, \mathit{imm}}(I)$ is a local best response given $I$ was reached. We now have all the notation we need for the pseudocode of PBE-CFR (Algorithm~\ref{alg:pbe_cfr}) and a sketch of all associated proofs (Section~\ref{sec:pbe_proof}; see App.~\ref{app:proofs} for details). 

\begin{algorithm}[htp!]
\small
\caption{\textsc{PBE-CFR}}
\label{alg:pbe_cfr}
\begin{algorithmic}[1]
\Require{Input game $G$, number of timesteps $T$}

\For{$I \in G$}
\State{$j = V(I)$}
\State{$\bsigma^1(I)(a) \gets \frac{1}{\vert A_{j}(I) \vert}$ for all $a \in A_{j}(I)$}
\State{$\mu(h \given I) \gets \frac{1}{\vert I(h) \vert}$ for all $h \in I$}
\State{Initialize $R^T_{j, \mathit{imm}}(I)(a) \gets 0$ for all $a \in A_{j}(I)$}
\State{Initialize cumulative infoset strategy weights $S_I(a) \gets 0$ for all $a \in A_{j}(I)$}
\State{Initialize $U^E(\cdot \given h) = 0$ for all $h \in I$}
\State{Initialize $U^B(\cdot \given I) = 0$}
\EndFor

\For{$t \in \{1, \dotsc, T\}$}
\State{$U^E(\bsigma^t \given \emptyset) \gets \textsc{TraverseWithBeliefs}\left( G, \emptyset, U^E, \bm{1}_{3}, \bsigma^t, \mu^t \right)$} 
\State{$\mu \gets \textsc{UpdateBeliefs}(G, \bsigma^{t + 1})$}
\EndFor
\For{$I \in G$}
\State{$\bsigma^{*}(I) \gets \textsc{Average}\left( \{\bsigma^t(I)\}_{t = 1}^T \right)$}
\EndFor
\State{$\mu^{*} \gets \textsc{UpdateBeliefs}(G, \bsigma^{*})$}\\
\Return $\bsigma^{*}, \mu^{*}$
\end{algorithmic}
\end{algorithm}


We now describe the scheme of PBE-CFR in terms of two major but natural modifications to the original CFR algorithm \citep{cfrm07}. First, in CFR, the counterfactual regrets of player $j$'s strategy at information set $I$ are weighted by the probability that $I$ was reached by $\bsigma_{-j}$, given that player $j$ played to reach $I$. Furthermore, when computing the average strategy for $I$ at the end of CFR, every strategy $\bsigma^t_j(I)(a)$ is weighted by the likelihood $r_j(\bsigma^t, I)$ of that state being reached by $I$. Instead in PBE-CFR, we compute the believed utility $U^B(\bsigma, \mu \given I)$ at every information set $I$ given strategy profile $\bsigma$ and belief system $\mu$ (Definition~\ref{def:seq_rat}), \textit{given that it was reached} by $\bsigma$. Hence, we exclude the aforementioned probability of reaching $I$ associated with $\bsigma_{-j}$ as part of $U^B(\bsigma, \mu \given I)$ and also $r_j(\bsigma^t, I)$ at the end when computing the average strategy. Moreover, the immediate believed regret $R^T_{j, \mathit{imm}}(I)$ is computed cumulatively using the strategy $\bsigma^t$ at timestep $t$, the belief that node $h \in I$ has been reached $\mu^t(h \given I)$ at timestep $t$, and the expected utility of taking each action $a \in A(I)$ at node $h$, $U^E_j\left(\left.\bsigma^t\right|_{I \rightarrow a} \given \mathit{ha}\right)$. $U^E_j$ is computed separately during recursive calls to \textsc{TraverseWithBeliefs} (Algorithm~\ref{alg:pbe_traverse}).


The second change is that after updating $\bsigma$ for timestep $t+1$ and returning from the original call to \textsc{TraverseWithBeliefs}, $\mu$ is also updated for the next timestep using \textsc{UpdateBeliefs} (Algorithm~\ref{alg:update_beliefs}). \textsc{UpdateBeliefs} first constructs a plausibility order $\order$ given $\bsigma^{t+1}$ and then computes $\mu^{t+1}$ for each information set $I$ as follows.
If $I$ is off of the equilibrium path, the nodes of $I$ are divided into two tiers according to their relative plausibilities in $\order$, with the most plausible nodes being added to set $V$. $\mu^{t+1}(\cdot \given I)$ is set to the uniform distribution over all nodes in $V$ and 0 for all nodes excluded from $V$. If $I$ is on the equilibrium path, meaning $r(I, \bsigma^{t+1}) > 0$, then $\mu^{t+1}$ is updated using Bayes' rule and the reach probabilities of each node in $I$ given $\bsigma^{t+1}$.



\section{Theoretical results}\label{sec:pbe_proof}

Our first result establishes that the space and time complexity of PBE-CFR is polynomial as a function of the input game size and the number of timesteps $T$. The proof is relegated to App.~\ref{app:proof_pbe_complex}.
\begin{theorem}\label{thm:pbe_complex}
The worst-case space and time complexities of PBE-CFR are $O(\lvert H \rvert \cdot \lvert A_{max} \rvert^2)$ and $O(T \cdot \lvert H \rvert \cdot \lvert A_{max} \rvert^2)$ respectively, where $A_{max}$ is the largest action set across all players' information sets.
\end{theorem}

Next, we prove that PBE-CFR is guaranteed to converge to a PBE for two-player zero-sum EFGs. We will use the concept of \term{local sequential rationality} which means that the property of sequential rationality at information set $I$ holds for all strategies that differ from $\bsigma$ only at $I$. \citet{hendon94} state that if beliefs $\mu$ are consistent, we need only consider these local deviations
at each information set $I$ in order to verify sequential rationality for $I$. The one-shot deviation principle follows:
\begin{definition}[One-shot deviation]\label{def:one_shot}
    Let $(\bsigma, \mu)$ be an assessment that satisfies \term{local sequential rationality} at every information set, meaning that   
    the property of sequential rationality at information set $I$ holds for all strategies that differ from $\bsigma$ only at $I$. If $(\bsigma, \mu)$ is also consistent, then $(\bsigma, \mu)$ is sequentially rational and therefore a sequential equilibrium.
\end{definition}
We prove that the assessment $(\bsigma^{*}, \mu^{*})$ returned by PBE-CFR satisfies sequential rationality at every player information set.
\begin{theorem}\label{thm:PBECFR}
In a two-player zero-sum game, for any information set $I \in \cI_j$ , $j \neq 0$, a consistent assessment $(\bsigma^{*}, \mu^{*})$, and any strategy profile $\sigma'_j \in \Sigma_j$,
\begin{equation*}
U^B_j\left(\sigma'_j \bsigma^{*}_{-j}, \mu^{*} \given I\right) \leq U^B_j\left(\bsigma^{*}, \mu^{*} \given I\right).
\end{equation*}
\end{theorem}

We break the proof down into lemmas and provide all omitted proofs of these lemmas in Appendices~\ref{app:proof1} through \ref{app:proof_imm_believed}. 
We first demonstrate that the immediate believed regret at any information set $I$ after running PBE-CFR for $T$ timesteps, given by $R^T_{j, \mathit{imm}}(I)$, is equal to the immediate believed regret of the average strategy $\bsigma^{*}$ given a consistent belief $\mu^{*}$.

\begin{lemma}\label{lemma:agm}
    $(\bsigma^{*}, \mu^{*})$ is an AGM-consistent assessment rationalized by plausibility order $\order$, and $\mu$ is Bayesian relative to $\order$.
\end{lemma}

Absent the algorithm, the immediate believed regret of the returned assessment $(\bsigma^{*}, \mu^{*})$ at information set $I$ is given by
\begin{equation*}
    R^{*}_{j, \mathit{imm}}(I) = \max_{a \in A_j(I)} U^B_j\left(\mu^{*}, \left.\bsigma^{*}\right|_{I \rightarrow a} \given I\right) - U^B_j\left(\mu^{*}, \bsigma^{*} \given I\right).
\end{equation*}

If the immediate believed regret after $T$ timesteps at information set $I$ given the strategy $\bsigma^t$ and belief $\mu^t$ at each timestep $t$ can be written in accordance with the domain of regret-matching, Blackwell's approachability theorem applies, and convergence is guaranteed for two-player zero-sum games. In a zero-sum game, the range of utilities to player $j$ is $\Delta_{u, j} = \max\limits_{z \in Z} u_j(z) - \min\limits_{z \in Z} u_j(z)$; given this range, we have the following lemma for convergence:
\begin{lemma}\label{lemma:imm_bel_reg}
For any information set $I \in \cI_j$ in a two-player zero-sum game, where $R^{*}_{j, \mathit{imm}}(I)$ denotes the immediate believed regret of the average strategy $\bsigma^{*}$ given belief $\mu^{*}$ at $I$ and $R^T_{j, \mathit{imm}}(I)$ denotes the cumulative immediate believed regret at $I$ after $T$ timesteps,
\begin{equation*}
    R^{*}_{j, \mathit{imm}}(I) \leq R^T_{j, \mathit{imm}}(I) \leq \varepsilon,
\end{equation*}
satisfying local sequential rationality for large enough $T$ where
\begin{equation*}
    T \leq \left( \frac{\Delta_{u, j} \vert A_j(I) \vert}{\varepsilon} \right)^2.
\end{equation*}
\end{lemma}
We now show that the one-shot deviation principle is satisfied, completing the proof of Theorem~\ref{thm:PBECFR}.
\begin{lemma}\label{lemma:localBR}
    For a given finite EFG $G$, any player $j$, and a consistent assessment $(\mu^{*}, \bsigma^{*})$ learned through \textsc{PBE-CFR}, if $\pi'_j = \{ a \in \argmax\limits_{a \in A(I)} R^*_{j, \mathit{imm}}(I) \}_{I \in \cI_j}$, then $\pi'_j$ is a sequential best response to $(\mu^{*}, \bsigma^{*})$ $\iff$ $\pi'_j(I)$ is a local best response to $(\mu^{*}, \pi'_j, \bsigma^{*}_{-j})$ for all $I \in \cI_j$.
\end{lemma}

\begin{lemma}\label{lem:imm_believed}
    If local sequential rationality is satisfied at every information set by strategy $\pi'_j$, then the consistent assessment $(\bsigma^{*}, \mu^{*})$ is also sequentially rational, with $R^{*}_{j, \mathit{imm}}(I) \leq \frac{\Delta_{u, j} \vert A_j(I) \vert}{\sqrt{T}}$ at every information set.
\end{lemma}


\section{Experiments}\label{sec:pbe_exp}
\subsection{Experimental Setup}\label{sec:setup}
\begin{figure}[h]
    \centering
    \includegraphics[width=0.9\columnwidth]{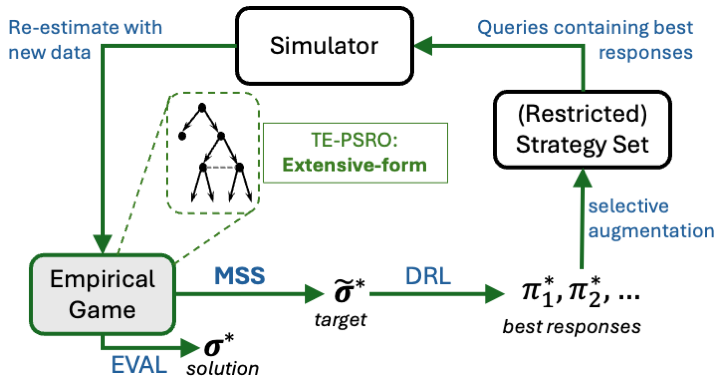}
    \caption{TE-PSRO Schematic: Empirical game is extensive-form, so PBE may be used as MSS and/or EVAL.}
    \label{fig:TEPSRO_schematic}
\end{figure}
We begin with an overview the TE-PSRO framework \citep{konicki22,konicki25} and the two parameterized classes of general-sum imperfect-information games, \gengoof and \dondG, which we use in our experiments. 


\textbf{Policy space response oracles (PSRO)} \citep{psro17} is a powerful implementation of the empirical game-theoretic analysis (EGTA) approach \citep{wellman2025empirical}. 
Given access to a simulator that encapsulates the full procedural description of a prohibitively complex game (called the \term{true game}), EGTA uses accumulated simulation data to induce a coarser but more tractable model of the game called the \term{empirical game}.
The empirical game generally covers a restricted space of the original strategy profile space. 
In iterative approaches to EGTA, analysis of the empirical game drives further refinement of the model through extension of the profile space.

In PSRO, the model is updated in the following steps, illustrated in Figure~\ref{fig:TEPSRO_schematic}). First, an arbitrary game-solving algorithm, called the meta-strategy solver (MSS) in this context, is applied to the current empirical game to obtain a solution called the \emph{target}. Then, each agent's best response (BR) to this target is approximated by a single-agent deep reinforcement learning (DRL) approach in an environment consistent with the simulator. Finally, agents' strategy sets are augmented with the respective BRs, and the simulator is queried to obtain further data (payoff information, in particular) to complete the latest empirical game iterate. Moreover, a game-solver termed EVAL, not necessarily the same as the MSS, gauges model quality and decides whether further refinement is necessary. 

In prior work  \cite{konicki22}, we introduced the Tree-Exploiting PSRO (TE-PSRO) variant where the empirical game is in extensive form, though still a coarser version of a full description of the true game. We followed up with methodological advances to improve the tradeoff between tractability and fidelity of the induced empirical game, and hence the scalability of TE-PSRO for imperfect-information games \cite{konicki25}. We devised an \emph{abstraction framework} in which each edge of the extensive-form empirical game represents a DRL-derived implicit \emph{policy} executable in the simulator, allowing much of the underlying state and observation spaces to remain implicit in the model. We also employed a parameterized heuristic to control the growth of the empirical game tree by adding edges induced by the latest BR policies at select information sets only. For a fixed integer $M$, we first estimated the gains of playing BR policies rather than target policies at candidate information sets of the current model, constructed a  softmax distribution over these information sets using those gains, and then sampled (up to) $M$ information sets from this distribution for adding edges to. We adopt this framework in this paper too and call $M$ the \term{growth parameter}.

\textbf{\gengoofK{K}} \citep{konicki25}, parameterized by a positive integer $K > 1$, generalizes the $2$-player version of the widely studied symmetric zero-sum card game Goofspiel \citep{rhoads_goof12} to $K-1$ rounds and arbitrary real-valued utilities. We start with a support of $K$ possible stochastic outcomes and a categorical distribution over them sampled uniformly at random. At the start of each round, Nature uniformly samples one outcome without replacement, re-normalizing the distribution over the residual support for the next round; then, players~$1$ and $2$ sequentially choose one of $K$ respective actions each, observing the full history of all previous rounds and the latest revealed stochastic outcome. For each triplet of stochastic outcome and players' actions, a uniformly random finite reward is sampled for each player and publicly revealed; the utility of each player on termination is the sum of rewards over all rounds. 

Additionally, we introduce a novel modified version of this game class called \textbf{\modggK{K}} which differs from \gengoofK{K} in the following way only: in each round, player~$2$ observes player~$1$'s action before moving but neither player observes the revealed stochastic outcome, the history of past rounds still being public. If the true game is \modgg, we tend to have more non-singleton information sets in empirical games in TE-PSRO iterations than those for \gengoof.

\textbf{\dondG} \citep{konicki25} is a finite-horizon negotiation game where two players engage in an alternating-offer bargaining protocol to decide how to split a public pool of indivisible items of multiple types between themselves. Each player has a vector of private valuations over item types, satisfying mild assumptions, as well as an \emph{outside offer} in the form of a private set of items of the same types. At the start of the game, Nature picks valuation vectors and outside offers from public probability distributions. Each player is also allowed to communicate to the other a binary signal (high/low) indicating whether the value of their private offer exceeds a fixed threshold; we encode the decision of whether or not to disclose this coarsened information by another binary signal (true/false) called \textit{revelation}.
The game proceeds in rounds, with players~$1$ and~$2$ sequentially taking one action each from the following options in each round: accept the other player's latest offer (if any), walk away (ending the game), or produce an offer-revelation combination. An offer takes the form of a proposed partition of the pool between the agents. If an offer is accepted by the other player in any round, the game stops, the pool is split accordingly, and each agent's realized utility is the total value of their share in the split; otherwise, negotiation fails and each agent receives their outside offer. Each agent's utility is geometrically discounted over rounds.

Detailed descriptions of these game classes along with additional references and respective parameter configurations used in our experiments are available in Appendices \ref{app:games_gengoof} (\gengoof), \ref{app:games_abs_private} (\modgg) and \ref{app:games_bargain} (\dondG).

\subsection{PBE-CFR Performance Evaluation}\label{sec:pbe_cfr_time_exp}
In our first set of experiments, we estimated the effectiveness of PBE-CFR in approximating a PBE of a general-sum imperfect-information game as well as the memory and wall time needed for convergence. We generated test games of varying complexity by running multiple iterations of TE-PSRO (which we call epochs to distinguish them from CFR/PBE-CFR iterations) on several parameterized instances of \modggK4 and \modggK5; in each epoch, we used deep Q-networks for best-response approximation, using the same methodology as \citet{konicki25}. This resulted in approximately $1200$ empirical games for \modggK4 and approximately $800$ for \modggK5 across all epochs. $2$ and $3$ GB of memory were sufficient for completing every full TE-PSRO run for \modggK4 and \modggK5 respectively on our local computing cluster using a single core.

We gauged the approximation quality of PBE-CFR by measuring how close it gets to achieving local sequential rationality, which implies sequential rationality by the one-shot deviation principle (Section~\ref{sec:pbe_proof}). Note that the solution generated by PBE-CFR satisfies the other two defining criteria of PBE by construction. We applied PBE-CFR to each of our \modgg empirical games with different values of the total number of PBE-CFR iterations $T$. For each solution, we computed the regret at each information set of not choosing another action $a \in A_j(I)$, given the assessment $(\bsigma^*, \mu^*)$, and recorded the maximum of all these regrets, termed the \term{worst-case local regret}. Table~\ref{table:seq_rat_pbe} shows the resulting worst-case local regrets, averaged over all empirical games for each \modgg variant: for all $T$, we obtain regret values of the order of $10^{-3}$ to $10^{-2}$ for leaf utilities of the order of $10^1$, with a slight reduction as $T$ increases. This suggests that PBE-CFR closely approximates a PBE of these general-sum games. 

\begin{table}[htp]
\centering
\begin{tabular}{ |c||c|c| } 
\hline
  $T$ & \modggK4 & \modggK5 \\ 
 \hline 
 500 & 0.0104 & 0.0113 \\
 \hline
 1000 & 0.0080 & 0.0099 \\
 \hline
2000 & 0.0078 & 0.0097 \\
\hline
5000 & 0.0073 & 0.0096 \\
\hline
\end{tabular}
\caption{Worst-case local regret of PBE in \modggK4\ and \modggK5\ for various values of $T$.}
\label{table:seq_rat_pbe}
\end{table}

To assess speed, we applied traditional CFR with the same values of $T$ to each empirical game in parallel to PBE-CFR, and recorded the respective running times. Figure~\ref{fig:pbe_runtimes} provides scatter plots of these running times against the sizes of the corresponding games measured in terms of the total number of information sets of both players for a representative value of $T$. PBE-CFR running times are typically larger than but of the same order of magnitude as those for CFR. The slowdown is reasonable given the additional modules that PBE-CFR needs to execute to ensure equilibrium refinement. Plots for other values of $T$, being qualitatively similar, are omitted.


\begin{figure}[htp]
\centering
\begin{subfigure}[b]{0.45\textwidth}
\centering
    \includegraphics[width=0.75\columnwidth]{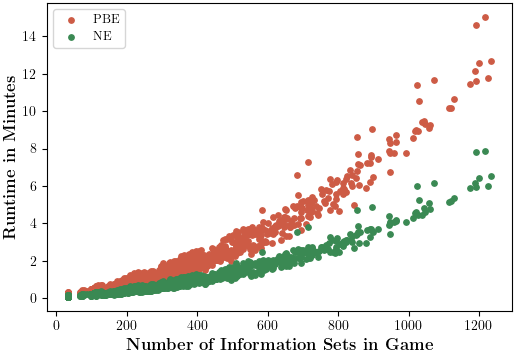}
\end{subfigure}\\
\vspace{10pt}
\begin{subfigure}[b]{0.45\textwidth}
\centering
    \includegraphics[width=0.75\columnwidth]{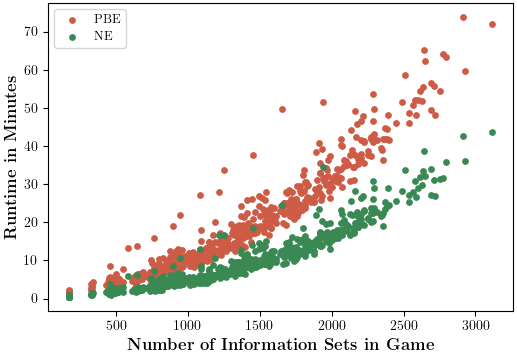}
\end{subfigure}
\caption{Time required by CFR and PBE-CFR for games generated from \modggK4\ (top) and \modggK5\ (bottom) with $T = 1000$.}
    \label{fig:pbe_runtimes}
\end{figure}

\subsection{Application to TE-PSRO as MSS}\label{sec:pbe_exp_tepsro}\label{sec:tepsro_expts}
We conducted another set of experiments to evaluate the advantage that may be gained by using PBE computed using PBE-CFR as the MSS in TE-PSRO. We used traditional CFR, which approximates NE with no guarantee of refinement, as a baseline MSS for the same true game(s) and parameter configurations.   
We drew multiple true game instances from the \dondG and \gengoofK4 classes and used several values of the growth parameter $M$ from $\{1, 2, 4, 8, 16\}$. We set $T=500$ for both CFR and PBE-CFR. 
Additionally for \gengoofK4, we experimented with different degrees of coarseness of the empirical games by specifying which rounds' stochastic event could be included in the empirical game tree, which we denote by $IR$ to refer to \term{included rounds}, zero-indexed with respect to the root; e.g., $IR=[0,1]$ means that the third and last stochastic event in the true game is necessarily abstracted away from every TE-PSRO-induced empirical game by construction. 
For each empirical game, we used the NE returned by CFR as the EVAL regardless of the MSS and used the regret of this EVAL, computed with respect to the true game, as the metric of model quality. This is the regret that we plot on the vertical axis in Figures \ref{fig:barg_pbe_mss_chap8} and~\ref{fig:abstract_pbe_mss_chap8}.

Figure~\ref{fig:barg_pbe_mss_chap8} shows how regret varies over TE-PSRO epochs for \dondG, averaged over $25$ trials. Figure~\ref{fig:abstract_pbe_mss_chap8} shows the same for \gengoofK4, averaged over $5$ trials, for each of the $3$ IR treatments. Error bars correspond to a 95\% confidence interval. These figures correspond to two representative values of $M$ for each true game class; plots for all the values of $M$ we considered are available in App.~\ref{sec:app_pbe_mss}. Figure~\ref{fig:barg_pbe_mss_chap8} does not support a clear winner for \dondG: the regret curves stay close to each other while converging to approximately zero regret, with NE and PBE slightly outperforming the other as an MSS for $M = 4$ and $M = 8$ respectively. By contrast, in Figure~\ref{fig:abstract_pbe_mss_chap8}, PBE mostly appears to outperform NE more clearly as an MSS for \gengoofK4 under the same IR treatment. 
\begin{figure*}[t]
    \centering
    \begin{subfigure}[b]{0.44\textwidth}
    \includegraphics[width=0.85\textwidth]{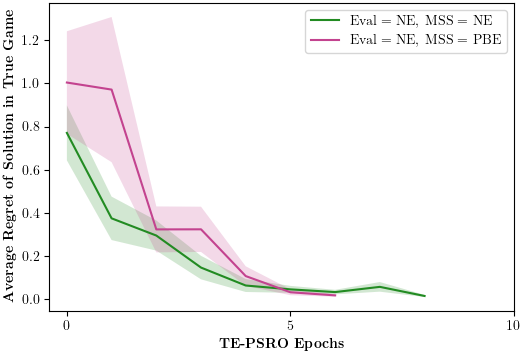}
    \caption{$M = 4$ \label{fig:pbe_barg_M4}}
    \end{subfigure}~
    \begin{subfigure}[b]{0.44\textwidth}
    \includegraphics[width=0.85\textwidth]{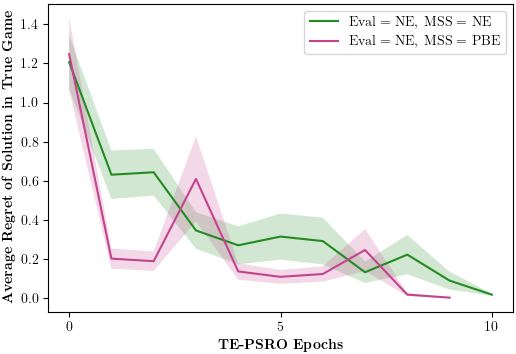}
    \caption{$M = 8$ \label{fig:pbe_barg_M8}}
    \end{subfigure}\\
    \caption{Average regret of $\bsigma^*$ evaluated in \dondG\ over the course of TE-PSRO's runtime, using NE or PBE as the MSS.}
    \label{fig:barg_pbe_mss_chap8}
\end{figure*}


\begin{figure*}[h]
    \centering
    \begin{subfigure}[b]{0.45\textwidth}
    \includegraphics[width=0.85\textwidth]{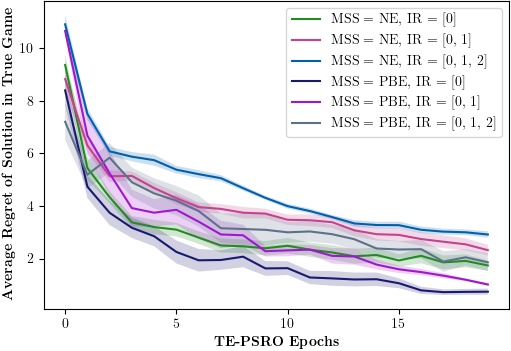}
    \caption{$M = 2$ \label{fig:pbe_abs4_M2}}
    \end{subfigure}~
    \begin{subfigure}[b]{0.45\textwidth}
    \includegraphics[width=0.85\textwidth]{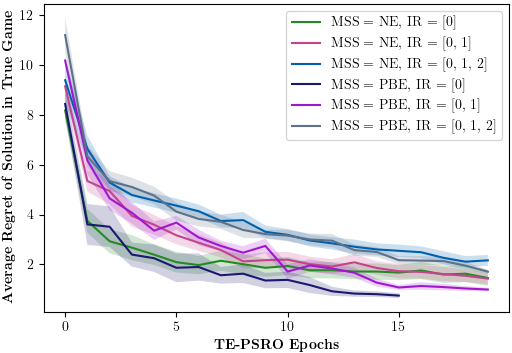}
    \caption{$M = 4$ \label{fig:pbe_abs4_M4}}
    \end{subfigure}\\
    \caption{Average regret of $\bsigma^*$ evaluated in \gengoofK4\ over the course of TE-PSRO's runtime, using NE or PBE as the MSS.}
    \label{fig:abstract_pbe_mss_chap8}
\end{figure*}

We offer a potential intuitive explanation for this difference in PBE performance between the two game classes, based on our observations of the structural evolution of the respective empirical game sequences over TE-PSRO epochs. For \gengoofK4, player 1's action in the current round is hidden from player 2. As $M$ increases, more information sets for both players have their action spaces augmented with new best response policies, leading to more non-singleton information sets belonging to player 2; this might be the reason why an MSS that incorporates beliefs for player 2 (PBE) is beneficial. For \dondG empirical games, imperfect information manifests only at the beginning of the game due to the opponent's outside offer signal being hidden, but this only persists as long as at least one agent keeps it signal hidden; thus, substantial portions of the empirical games for \dondG\ ended up containing primarily singleton information sets, rendering a refined MSS less useful.

\section{Discussion}\label{sec:disc}

We proposed the first algorithm that efficiently and effectively approximates a general PBE concept for arbitrary two-player EFGs of imperfect information.
Our algorithm specifically addresses the PBE concept defined by \citet{agm11}.
It is based on two non-trivial modifications to the classic CFR algorithm, which approximates unrefined NE.
 
 Given the ability to compute PBE, we investigate the opportunity to employ PBE for strategy exploration, as the MSS in a tree-exploiting variant of PSRO.
 We conduct experiments on two parameterized game classes, a general-sum variant on the card game Goofspiel, and a bargaining game with signaling options.
 We assess effectiveness in terms of the rate of convergence to low-regret empirical games, compared to unrefined NE as MSS. 
 Our results suggest that the benefit of PBE-as-MSS can depend significantly on structural properties of the game concerned. 
 In particular, we found the performance of PBE-as-MSS to be better for Goofspiel than for our bargaining game, as empirical game trees for the former tended to have more downstream non-singleton information sets.

 Natural future research directions include assessing PBE as an MSS for other game classes (e.g., poker) and improvements to PBE-CFR by invoking variants of CFR (Section~\ref{sec:relwork}).



\begin{acks}
This work was supported in part by the US National Science Foundation under CRII Award 2153184.
\end{acks}



\clearpage
\bibliographystyle{ACM-Reference-Format} 
\bibliography{references}


\clearpage
\onecolumn
\appendix

\section*{TECHNICAL APPENDICES}

\section{Illustrative examples}\label{sec:examples}
\subsection{Example illustrating belief systems and assessments}\label{app:main_example}
Consider the three-player, imperfect-information EFG depicted in Figure~\ref{fig:assessments_chap7} where Player 1 has two information sets $I^1_1 = \{ \emptyset \}$ and $I^2_1 = \{ \mathit{DA}, \mathit{DB} \}$, Player 2 has two singleton information sets $I^1_2 = \{ U \}$ and $I^2_2 = \{ D \}$, and Player 3 has one information set $I^1_3 = \{ \mathit{UL}, \mathit{UR} \}$. We will use it to illustrate several concepts introduced in Section~\ref{sec:prelim}. 

\begin{figure}[htp]
    \centering
    \includegraphics[width=0.53\linewidth]{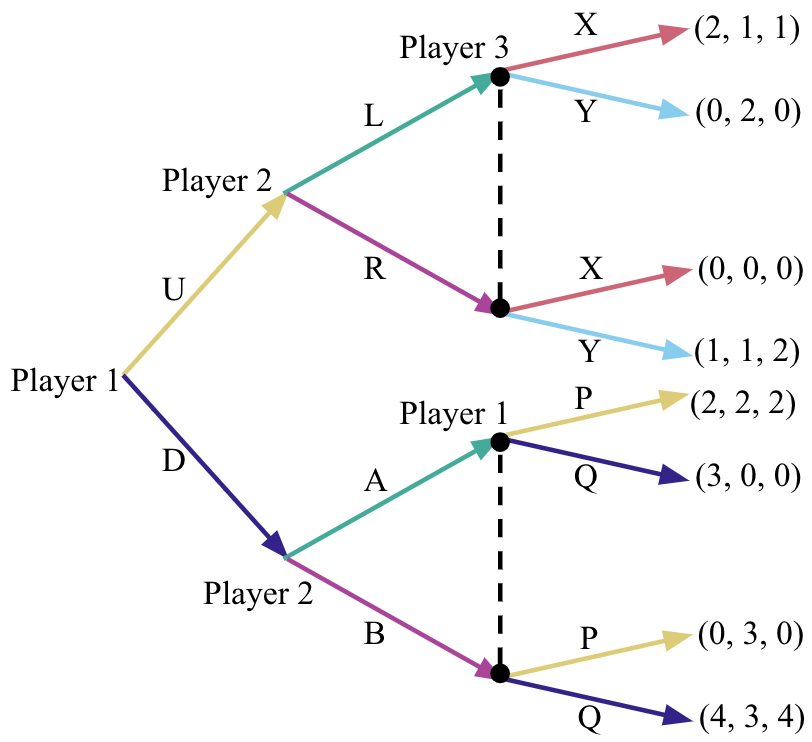}
    \caption{An imperfect information game for 3 players.}
    \label{fig:assessments_chap7}
\end{figure}

A possible strategy profile for this game is given by $\bsigma = (\sigma_1,\sigma_2,\sigma_3)$ defined as follows:
\begin{align*}
    &\sigma_1(I^1_1) = \{ U: \nicefrac{1}{3}, D: \nicefrac{2}{3} \},\ \sigma_1(I^2_1) = P; \\
    &\sigma_2(I^1_2) = \{ L: \nicefrac{1}{2}, R: \nicefrac{1}{2} \},\ \sigma_2(I^2_2) = B; \\
    &\sigma_3(I^1_3) = Y,
\end{align*}
where $\{a_1 : p_1, a_2 : p_2, \dots, a_m : p_m\}$ with $p_k \ge 0$ for every $k \in \{1,2,\dots,m\}$ and $\sum_{k=1}^m p_k=1$ represents a probability distribution over the set $\{a_1,a_2,\dots,a_m\}$; a denegerate distribution putting the entire probability mass on one action is represented by the action itself in a slight abuse of notation.

A possible system of beliefs $\mu$ for the same game is completely described by the following assignments: 
\begin{align*}
     &\mu(\emptyset \given I^1_1) = 1, \mu(\mathit{DB} \given I^2_1) = 1, \mu(\mathit{DA} \given I^2_1) = 0; \\
     & \mu(\mathit{U} \given I^1_2) = 1, \mu(\mathit{D} \given I^1_2) = 1; \\
     &\mu(\mathit{UL} \given I^1_3) = \mu(\mathit{UR} \given I^1_3) = \nicefrac{1}{2}.
\end{align*}
For the above $\bsigma$ and $\mu$, $(\bsigma,\mu)$ is a possible assessment of the game.

Let us describe salient aspects of the above assessment in plain language. Here, we assume that the players know each other's strategies. Because $\bsigma(I^1_1)$ is mixed, Player 3 believes that Player 1 will play $U$ with positive probability and hence he himself may have to move. Likewise, Player 1 knows that he may have to move a second time if he chooses to play $D$ with positive probability. If Player 3 must move, according to $\mu$, he assigns equal probability to him reaching history $\mathit{UL}$ and him reaching $\mathit{UR}$ during gameplay, given that information set $I^1_3$ has been reached. If Player 1 must move again, according to $\mu$, he believes with absolute certainty that history $\mathit{DB}$ has been reached and that history $\mathit{DA}$ will not, given that information set $I^2_1$ has been reached.
Incidentally, the assessment under consideration is AGM-consistent and compatible with Bayes' rule (see Section~\ref{sec:pbe}).

\subsection{Example of assessment violating AGM-consistency}\label{app:agm_example}
\begin{figure}[h]
    \centering
    \includegraphics[scale=0.5]{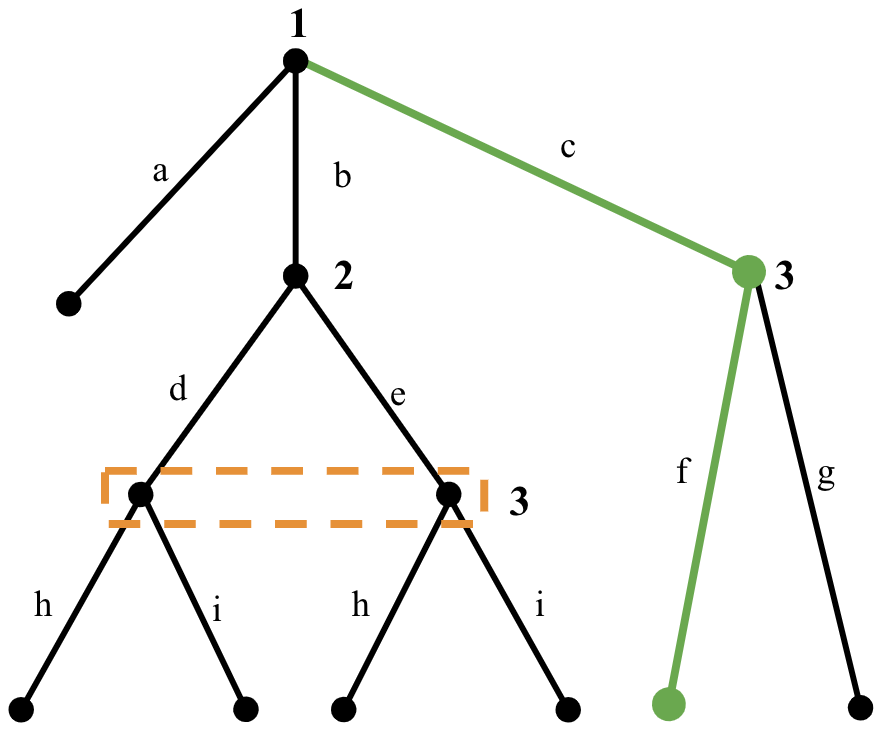}
    \caption{Reproduced example of an imperfect-information EFG from \citet{agm11} with one non-singleton information set (for Player~$3$) represented by the orange box. The equilibrium path induced by the assessment $(\bm{\sigma}, \mu)$  for which a viable plausibility order $\precsim$ cannot be constructed is highlighted in green.}
    \label{fig:not_agm_consist}
\end{figure}
Consider the three-player, imperfect-information EFG depicted in Figure~\ref{fig:not_agm_consist} and the following strategy profile for this game:
\begin{align*}
    \sigma_1(I(\emptyset)) &= c \\
    \sigma_2(I(b)) &= d \\
    \sigma_3(I(c)) &= f,\ \sigma_3(I(\mathit{bd})) = h.
\end{align*}

Using the definition of AGM-consistency given in Section~\ref{sec:pbe}, we can prove by contradiction that this strategy profile cannot produce an AGM-consistent assessment when paired with any system of beliefs $\mu$ with $\mu(\mathit{be}) > 0$, hence no such assessment can be a PBE for the game under consideration. Let us try to construct a plausibility order and check if it can be made AGM-consistent. Because $\sigma_2(d) = 1$, node $b$ must be as plausible as node $\mathit{bd}$ and more plausible than node $\mathit{be}$. It follows from transitivity that node $\mathit{bd}$ is more plausible than node $\mathit{be}$. But because Player 3's belief $\mu(\mathit{be})$ is strictly positive, node $\mathit{be}$ must be at least as plausible as node $\mathit{bd}$, the other node in the information set. This is a contradiction. 

\section{Omitted pseudocode and algorithm descriptions}\label{app:algo}
Here we provide all details of the algorithms sketched in Section~\ref{sec:alg_pbe}.

Algorithm~\ref{alg:verify_seq_rat} (\textsc{IsSequentRational}) checks whether the input assessment $(\bsigma, \mu)$ is sequentially rational (Definition~\ref{def:seq_rat}) for the input game $G$. It visits every player information set $I$ in the game for $j \in N \setminus \{ 0 \}$ and computes the believed utility $U^B_j\left(\bsigma, \mu \given I\right)$, given that $I$ is reached; it then iterates over the actions $a \in A_j(I)$, creating a new strategy profile $\bsigma'$ for each action which is identical to $\bsigma$ at every information set except $I$, where $\bsigma'(I) = a$; finally, it computes the difference between the believed utilities of $\bsigma'$ and $\bsigma$, returning \textsc{False} if $\bsigma'$ yields a higher payoff at $I$.

\begin{algorithm}[htbp!]
\small
\caption{\textsc{IsSequentRational}}
\label{alg:verify_seq_rat}
\begin{algorithmic}[1]
\Require{Game $G$, strategy profile $\bsigma$, belief system $\mu$}

\State{Acquire $\mathcal{I}_1, \mathcal{I}_2$ from input game $G$}
\For{$I \in \mathcal{I}_1 \cup \mathcal{I}_2$}
\State{$j \gets V(h \in I)$}
\State{$U^B_j\left(\bsigma, \mu \given I\right) \gets 0$}
\For{$h \in I$}
\For{$z \in Z$}
\State{$U^B_j\left(\bsigma, \mu \given I\right) \gets U^B_j\left(\bsigma, \mu \given I\right) + \mu(h \given I) \cdot r(z \given h, \bsigma) \cdot u_j(z)$}
\EndFor
\EndFor

\State{$\bsigma' \gets \mathsf{copy}(\bsigma)$}
\For{$a \in A(I)$}
\State{$\bsigma'(I)(a) \gets 1$}
\State{$U^B_j\left(\bsigma', \mu \given I\right) \gets 0$}
\For{$h \in I$}
\For{$z \in Z$}
\State{$U^B_j\left(\bsigma', \mu \given I\right) \gets U^B_j\left(\bsigma', \mu \given I\right) + \mu(h \given I) \cdot r(z \given h, \bsigma') \cdot u_j(z)$}
\EndFor
\If{$U^B_j\left(\bsigma', \mu \given I\right) > U^B_j\left(\bsigma, \mu \given I\right)$}
\Return \textsc{False}
\EndIf
\EndFor
\EndFor
\EndFor \\
\Return \textsc{True}
\end{algorithmic}
\end{algorithm}

Algorithm~\ref{alg:verify_bayes} (\textsc{SatisfiesBayes}) checks whether an assessment $(\bsigma, \mu)$ satisfies Bayes' rule at every information set of game $G$. It first ensures that $\mu$ specifies a proper probability distribution for a given information set $I$, meaning all components of the input map add up to 1; then, for every $I$ that is reachable by $\bsigma$, it computes the reach probability of every node $h \in I$ given that $I$ was visited and checks whether the probability is equal to $\mu(h \given I)$; any inequality leads \textsc{SatisfiesBayes} to return \textsc{False}.

\begin{algorithm}[htbp!]
\small
\caption{\textsc{SatisfiesBayes}}
\label{alg:verify_bayes}
\begin{algorithmic}[1]
\Require{Game $G$, strategy profile $\bsigma$, belief system $\mu$}

\State{Acquire $\mathcal{I}_1, \mathcal{I}_2$ from input game $G$}
\For{$I \in \mathcal{I}_1 \cup \mathcal{I}_2$}
\If{$\sum_{h \in I} \mu(h \given I) \neq 1$}
\Return \textsc{False}
\EndIf

\If{$r(I, \bsigma) > 0$}
\For{$h \in I$}
\If{$\frac{r(h, \bsigma)}{r(I, \bsigma)} \neq \mu(h \given I)$}
\Return \textsc{False}
\EndIf
\EndFor
\EndIf
\EndFor \\
\Return \textsc{True}
\end{algorithmic}
\end{algorithm}

Algorithm~\ref{alg:verify_agm} (\textsc{IsConsistent}) verifies AGM-consistency for $(\bsigma, \mu)$ given $G$ using two helper methods. Given the assessment, \textsc{ConstructOrderGivenProfile} constructs a plausibility order $\order$ using $\bsigma$ and the chance outcomes $X(\cdot)$ (included as part of $G$) according to the first two conditions for AGM-consistency outlined in Definition~\ref{def:agm}. Then, \textsc{UpdateOrderGivenBelief} is called to modify $\order$ according to $\mu$ one information set at a time according to the third and final condition of AGM-consistency. If any contradictions arise such as the example in Figure~\ref{fig:not_agm_consist}, \textsc{UpdateOrderGivenBelief} returns \textsc{None}, so \textsc{IsConsistent} returns \textsc{False}. 

\begin{algorithm}[htp!]
\small
\caption{\textsc{IsConsistent}}
\label{alg:verify_agm}
\begin{algorithmic}[1]
\Require{Game $G$, strategy profile $\bsigma$, belief system $\mu$}

\State{Acquire $\mathcal{I}_1, \mathcal{I}_2, X$ from input game $G$}
\State{$\order\ \gets \textsc{ConstructOrderGivenProfile}(\mathcal{I}_1 \cup \mathcal{I}_2, X, \bsigma)$}
\State{$\order\ \gets \textsc{UpdateOrderGivenBelief}(\mathcal{I}_1 \cup \mathcal{I}_2, \mu, \order)$}
\If{$\order$ is \textsc{None}}
\Return \textsc{False}
\EndIf \\
\Return \textsc{True}
\end{algorithmic}
\end{algorithm}

At each information set $I$ and every decision node $h \in I$, \textsc{ConstructOrderGivenProfile} checks which of the actions specified by the strategy $\bsigma(I)$ have positive probability. For action $a$, if the child node $ha$ is reached by $\bsigma(I)$, $\order$ is updated with a new binary relation between $h$ and $ha$, in accordance with Definition~\ref{def:agm}. Additionally, the child nodes of $h$ are grouped into set $V$ if they reached by $\bsigma(I)$ or set $W$ if not. Then for every $v \in V$ and $w \in W$, a new binary relation $v \prec w$ is added to $\order$ in accordance with transivity (Definition~\ref{def:plaus}). \textsc{ConstructOrderGivenProfile} also iterates over all the chance nodes $h \in X$; for every outcome $x \in X(h)$, a new binary relation $h \sim hx$ is also added to $\order$.

\begin{algorithm}[htp!]
\small
\caption{\textsc{ConstructOrderGivenProfile}}
\label{alg:build_order_agm}
\begin{algorithmic}[1]
\Require{List of player information sets $\mathcal{I}_1 \cup \mathcal{I}_2$, chance node distribution $X$, strategy profile $\bsigma$}

\State{$\order\ \gets \{ \}$}
\For{$I \in \mathcal{I}_1 \cup \mathcal{I}_2$}
\For{$h \in I$}
\State{$V, W \gets \emptyset$}
\For{$a \in \bsigma(I)$}
\If{$\bsigma(I)(a) > 0$}
\State{$\order\gets \order \cup \{ h \sim \mathit{ha} \}$}
\State{$V \gets V \cup \{\mathit{ha}\}$}
\Else
\State{$\order \gets \order \cup \{ h \prec \mathit{ha} \}$}
\State{$W \gets W \cup \{\mathit{ha}\}$}
\EndIf
\EndFor
\For{$(h_1, h_2) \in V \times W$}
\State{$\order \gets \order \cup \{ h_1 \prec h_2 \}$}
\EndFor
\EndFor
\EndFor
\For{$h \in X$}
\For{$x \in X(h)$}
\State{$\order \gets \order \cup \{ h \sim hx \}$}
\EndFor
\EndFor \\
\Return $\order$
\end{algorithmic}
\end{algorithm}

Given an $\order$ constructed so that assessment $(\bsigma, \mu)$ meets the first two conditions of AGM-consistency (Definition~\ref{def:agm}), \textsc{UpdateOrderGivenBelief} iterates over each non-singleton information set, since a belief distribution over a singleton information set is trivial. At information set $I$, all decision nodes in $I$ with positive belief $\mu(h \given I)$ are grouped into set $V$ and the remaining nodes are grouped into $W$. For any pairing $(h_1, h_2)$ of nodes in $V$, the subroutine requires that the pairing does not already have a relation specified in $\order$, or that $\order$ states that $h_1 \sim h_2$; if both conditions are not met, the subroutine returns \textsc{None}. \textsc{UpdateOrderGivenBelief} also checks every pairing $(h_1, h_2)$ with $h_1 \in V$ and $h_2 \in W$ to see if the pairing does not already have a relation specified in $\order$, or whether $\order$ states that $h_1 \prec h_2$; if both conditions are not met, the subroutine returns \textsc{None}.


\begin{algorithm}[htp]
\small
\caption{\textsc{UpdateOrderGivenBelief}}
\label{alg:update_order_agm}
\begin{algorithmic}[1]
\Require{List of player information sets $\mathcal{I}_1 \cup \mathcal{I}_2$, belief system $\mu$, plausibility order $\order$}

\For{$I \in \mathcal{I}_1 \cup \mathcal{I}_2$}
\State{$\textsc{Assert} \sum_{h \in I} \mu(h \given I) = 1$}
\If{$\textsc{len}(I) > 1$}
\State{$V = \{ h \mid h \in I, \mu(h \given I) > 0\}$}
\State{$W = I \setminus V$}
\For{$(h_1, h_2) \in \binom{V}{2}$}
\If{$(h_1, h_2) \in \order$}
\If{$h_1 \not\sim h_2$}
\Return \textsc{None}
\EndIf
\Else
\State{$\order \gets \order \cup \{ h_1 \sim h_2 \}$}
\EndIf
\EndFor
\For{$(h_1, h_2) \in V \times W$}
\If{$(h_1, h_2)$ is ordered in $\order$}
\If{$h_1 \not\prec h_2$}
\Return \textsc{None}
\EndIf
\Else
\State{$\order \gets \order \cup \{ h_1 \prec h_2 \}$}
\EndIf
\EndFor
\EndIf
\EndFor \\
\Return $\order$
\end{algorithmic}
\end{algorithm}

\section{Full description of PBE-CFR}\label{app:full_pbecfr}

PBE-CFR is an adaptation of CFR that minimizes the \term{believed regret} of playing $\bsigma$ at each information set given a belief system $\mu$ while keeping $\mu$ consistent. The \term{immediate believed regret} of playing $\bsigma$ at information set $I$ at timestep $T$ is therefore given by
\begin{align*}
    R^T_{j, \mathit{imm}}(I) &\coloneq \max\limits_{a \in A(I)} R^T_{j, \mathit{imm}}(I, a) \\
    &= \frac{1}{T} \max\limits_{a \in A(I)} \sum_{t=1}^T U^B_j \left( \mu^t, \left.\bsigma^t\right|_{I \rightarrow a} \given I \right) - U^B_j \left( \mu^t, \bsigma^t \given I \right) \\
    &= \frac{1}{T} \max\limits_{a \in A(I)} \sum_{t=1}^T \left( \sum_{h \in I} \mu^t(h \given I) \cdot \left( U^E_j\left(\left.\bsigma^t\right|_{I \rightarrow a} \given \mathit{ha}\right) - \sum_{a' \in A(I)} U^E_j\left(\bsigma^t \given \mathit{ha'}\right) \cdot  \bsigma^t(I)(a') \right) \right),
\end{align*}
where playing $a \in \argmax R^T_{j, \mathit{imm}}(I)$ is the local best response given $I$ was reached and where $U^B_j \left( \mu^t, \left.\bsigma^t\right|_{I \rightarrow a} \given I \right)$ denotes the \term{believed action utility} of playing any action $a$ at $I$. $D(I)$ denotes the information sets of $j$ reachable from $I$ and $\left.\bsigma\right|_{D(I) \rightarrow \pi'}$ denotes a strategy profile equal to $\bsigma$ except at the information sets of $D(I)$ where it is equal to the pure strategy $\pi'$. The \term{full believed regret} is
\begin{align*}
    R^T_{j, \mathit{full}}(I) &= \frac{1}{T} \max\limits_{\pi'_j \in \Pi_j} \sum_{t=1}^T U^B_j \left( \mu^t, \left.\bsigma^t\right|_{D(I) \rightarrow \pi'_j} \given I \right) - U^B_j \left( \mu^t, \bsigma^t \given I \right).
\end{align*}

\begin{algorithm}[htp!]
\small
\caption{\textsc{PBE-CFR}}
\label{alg:pbe_cfr_repeat}
\begin{algorithmic}[1]
\Require{Input game $G$, timesteps $T$}

\For{$I \in G$}
\State{$j = V(I)$}
\State{$\bsigma^1(I)(a) \gets \frac{1}{\vert A_{j}(I) \vert}$ for all $a \in A_{j}(I)$}
\State{$\mu(h \given I) \gets \frac{1}{\vert I(h) \vert}$ for all $h \in I$}
\State{Initialize cumulative immediate believed regrets $R^T_{j, \mathit{imm}}(I)(a) \gets 0$ for for all $a \in A_{j}(I)$}
\State{Initialize cumulative infoset strategy weights $S_I(a) \gets 0$ for all $a \in A_{j}(I)$}
\State{Initialize expected utilities $U^E(\cdot \given h) = 0$ for all $h \in I$}
\State{Initialize believed utilities $U^B(\cdot \given I) = 0$}
\EndFor

\For{$t \in \{1, \dotsc, T\}$}
\State{$U^E(\bsigma^t \given \emptyset) \gets \textsc{TraverseWithBeliefs}\left( G, \emptyset, U^E, \bm{1}_{3}, \bsigma^t, \mu^t \right)$} 
\State{$\mu \gets \textsc{UpdateBeliefs}(G, \bsigma^{t + 1})$}
\EndFor
\For{$I \in G$}
\State{$\bsigma^{*}(I) \gets \textsc{Average}\left( \{\bsigma^t(I)\}_{t = 1}^T \right)$}
\EndFor
\State{$\mu^{*} \gets \textsc{UpdateBeliefs}(G, \bsigma^{*})$}\\
\Return $\bsigma^{*}, \mu^{*}$
\end{algorithmic}
\end{algorithm}

The original CFR algorithm traverses the tree recursively, updating strategies and counterfactual regret at each information set and returning the average strategy which minimizes regret for the entire game. The average strategy returned by the algorithm converges to NE and is the extensive-form equivalent of the mixed NE that would be returned if the input game were in normal form. Therefore, the counterfactual regrets of player $j$'s strategy at information set $I$ must be weighted by the probability that $I$ was reached by $\bsigma_{-j}$, given that player $j$ played to reach $I$. Furthermore, when computing the average strategy for $I$ at the end of CFR, every strategy $\bsigma^t_j(I)(a)$ must be weighted by the likelihood of that state being reached by $I$, $r_j(\bsigma^t, I)$. Although this works when the goal is to minimize regret in the overall game, two key modifications to CFR are required in order to compute PBE since the goal is to minimize regret at every decision point in the game, alongside incorporating beliefs that follow Bayes' rule into this regret.

The first change is computing the believed utility $U^B(\bsigma, \mu \given I)$ at every information set $I$ given strategy profile $\bsigma$ and belief system $\mu$ (Definition~\ref{def:seq_rat}), \textit{given that it was reached} by $\bsigma$. This means not including the aforementioned probability of reaching $I$ associated with $\bsigma_{-j}$ as part of $U^B(\bsigma, \mu \given I)$, as defined earlier. Additionally, this means not incorporating $r_j(\bsigma^t, I)$ at the end when computing the average strategy since the algorithm minimizes regret at every information set under the assumption that it was reached. It is also important to mention that the immediate believed regret $R^T_{j, \mathit{imm}}(I)$ is computed cumulatively using the strategy $\bsigma^t$ at timestep $t$, the belief that node $h \in I$ has been reached $\mu^t(h \given I)$ at timestep $t$, and the expected utility of taking each action $a \in A(I)$ at node $h$, $U^E_j\left(\left.\bsigma^t\right|_{I \rightarrow a} \given \mathit{ha}\right)$. Thus, in order to compute $U^B$ at each information set, $U^E_j$ is computed separately during recursive calls to \textsc{TraverseWithBeliefs} (Algorithm~\ref{alg:pbe_traverse}).

The second change is that after updating $\bsigma$ for timestep $t+1$ and returning from the original call to \textsc{TraverseWithBeliefs}, $\mu$ is also updated for the next timestep using \textsc{UpdateBeliefs} (Algorithm~\ref{alg:update_beliefs}). \textsc{UpdateBeliefs} first constructs a plausibility order $\order$ given $\bsigma^{t+1}$ and then computes $\mu^{t+1}$ for each information set $I$, depending on whether $I$ is on or off the path of $\bsigma^{t+1}$ through $G$. If $I$ is off the equilibrium path, then the nodes of $I$ are divided into two tiers according to their relative plausibilities in $\order$, with the most plausible nodes being added to set $V$. $\mu^{t+1}(\cdot \given I)$ is set to the uniform distribution over all nodes in $V$ and 0 for all nodes excluded from $V$. If $I$ is on the equilibrium path, meaning $r(I, \bsigma^{t+1}) > 0$, then $\mu^{t+1}$ is updated using Bayes' rule and the reach probabilities of each node in $I$ given $\bsigma^{t+1}$.

\begin{algorithm}[htp]
\small
\caption{\textsc{TraverseWithBeliefs}}
\label{alg:pbe_traverse}
\begin{algorithmic}[1]
\Require{Input game $G$, current node $h$, expected utilities $U^E$, player reach probabilities vector $\bm{r}$, current strategy $\bsigma^t$, current beliefs $\mu^t$}

\State{$j \gets V(h)$}
\If{$h$ is terminal}
\Return{$u(h)$}
\ElsIf{$j = 0$}
\State{$\bm{r}' \gets \bm{r}$}
\State{$u^{E} = \vec{0}$}
\For{ $x \in X(h)$}
\State{$\bm{r}'_0 \gets \bm{r}_0 \cdot P(x \given h)$}
\State{$U^{E}(\mathit{hx}) \gets \textsc{TraverseWithBeliefs}\left(G, \mathit{hx}, U^E, \bm{r}', \bsigma^t, \mu^t \right)$}
\State{$u^E = u^E + P(x \given h) \cdot U^E(\mathit{hx})$}
\EndFor
\State{$U^E(h) = u^E$}

\Return $U^E$
\EndIf

\State{$I \gets I(h)$}
\State{Initialize immediate expected utilities $U^E_j(\bsigma^t \given \mathit{ha}) \gets 0$ for $a \in A_{j}(I)$}
\State{Initialize immediate believed action utilities $U^B_j(\bsigma^t, \mu^t \given I, a) \gets 0$ for $a \in A_{j}(I)$}
\For{$a \in A_{j}(I)$}
\State{$\bm{r}' \gets \bm{r}$}
\State{$\bm{r}'_j \gets \bm{r}_j \cdot \bsigma^t(I)(a)$}
\State{$U^E(\mathit{ha}) \gets \textsc{TraverseWithBeliefs}\left(G, ha, U^E, \bm{r}', \bsigma^t, \mu^t \right)$}
\State{$U^E(h) \gets U^E(h) + \bsigma^t(I)(a) \cdot U^E(\mathit{ha})$}
\State{$U^B_j(\bsigma^t, \mu^t \given I, a) \gets U^B_j(\bsigma^t, \mu^t \given I, a) + \mu^t(h \given I) \cdot U^E_j(\mathit{ha})$}
\State{$U^B_j(\bsigma^t, \mu^t \given I) \gets U^B_j(\bsigma^t, \mu^t \given I) + \bsigma^t(I)(a) \cdot U^B_j(\bsigma^t, \mu^t \given I, a)$}
\EndFor

\For{$a \in A_{j}(I)$}
\State{$R^T_{j, \mathit{imm}}(I, a) \gets R^T_{j, \mathit{imm}}(I, a) + U^B_j(\bsigma^t, \mu^t \given I, a) - U^B_j(\bsigma^t, \mu^t \given I)$}

\EndFor
\State{Update $\bsigma^{t + 1}(I)$ using $R^T_{j, \mathit{imm}}(\cdot)$ values and regret-matching}
\State{Update $S_I(a) \gets S_I(a) + \bsigma^{t + 1}(I)(a)$ for all $a \in A_j(I)$}\\

\Return $U^E(h)$
\end{algorithmic}
\end{algorithm}

\begin{algorithm}[htp]
\small
\caption{\textsc{UpdateBeliefs}}
\label{alg:update_beliefs}
\begin{algorithmic}[1]
\Require{Input game $G$, current strategy profile $\bsigma$}
\State{$\mu \gets \emptyset$}
\State{$\order\ \gets \textsc{ConstructOrderGivenProfile}(G, \bsigma)$}
\For{$I \in G$}
\State{$r(I, \bsigma) \gets \sum_{h \in I} r(h, \bsigma)$}

\If{$r(I, \bsigma) = 0$}
\State{$V \gets \emptyset$}
\For{$(h_1, h_2) \in \binom{I}{2}$}
\If{$(h_1, h_2) \in \order$ and $h_1 \prec h_2$}
\State{$V \gets V \cup \{h_1\}$}
\EndIf
\EndFor
\For{$h \in I$}
\If{$h \in V$}
\State{$\mu(h \given I) \gets \frac{1}{\vert V \vert}$}
\Else
\State{$\mu(h \given I) \gets 0$}
\EndIf
\EndFor

\Else
\State{$\mu(h \given I) \gets \frac{r(h, \bsigma)}{r(I, \bsigma)}$}
\EndIf
\EndFor \\

\Return $\mu$
\end{algorithmic}
\end{algorithm}

\section{Omitted proofs}\label{app:proofs}

\subsection{Proof of Theorem~\ref{thm:pbe_complex}}\label{app:proof_pbe_complex}
\begin{proof}
PBE-CFR requires space for the current strategy of each information set, the cumulative regrets of each information set (in this case, believed regrets), the cumulative sum of strategy weights for each information set, 
the beliefs at each information set, the expected utilities at each node, and the believed utilities at each information set at timestep $t$. These objects collectively take up $O\left( \vert H \vert \right)$ space.

During the call to \textsc{UpdateBeliefs} at timestep $t$, an order $\order$ must be constructed given $\bsigma$. The relation is constructed between each node $h$ in an information set $I$ and each of its children $ha$ given each action specified by $\bsigma(I)$, dividing its children between sets $V$ and $W$. Therefore, the total number of relations added to $\order$ through each node $h \in I$ is $\vert V \times W \vert$ + $\vert A_j(I) \vert$. However, $\vert W \vert$ is equal to $\vert A_j(I) \vert - \vert V \vert$, and 
\begin{equation*}
    \vert V \times W \vert \ = \ \vert V \vert \cdot \vert W \vert = \left( \vert A_j(I) \vert - \vert V \vert \right) \cdot \vert V \vert.
\end{equation*}

In the worst-case, $\vert V \vert = \vert W \vert$, meaning $\vert V \times W \vert = \nicefrac{1}{4} \cdot \vert A_j(I) \vert ^2$. Let $A_{max}$ denote the largest action space in the game across all players' information sets. The total space required by $\order$ in the worst case is therefore $O \left( \vert H \vert \cdot \vert A_{max} \vert ^2 \right)$.

In each of $T$ iterations, PBE-CFR must traverse the entire tree, updating cumulative believed regrets, utilities, and information set strategies. This portion of the algorithm over all $T$ steps requires time $O \left( T \cdot \vert H \vert \right)$. During each call to \textsc{UpdateBeliefs}, each information set is visited, and a new order $\order$ is constructed, as described earlier, visiting each node in the game tree once and dividing its children between two sets. This last step of \textsc{ConstructOrderGivenProfile} takes takes time $O \left( \vert A_{max} \vert ^2 \right)$ for a single node. It follows that this subroutine runs in $O \left( \vert H \vert \cdot \vert A_{max} \vert ^2 \right)$, and the total runtime of $T$ calls to \textsc{UpdateBeliefs} is $O( \left( T \cdot \vert H \vert \cdot \vert A_{max} \vert ^2 \right)$. Therefore, the runtime of PBE-CFR is $O( \left( T \cdot \vert H \vert \cdot \vert A_{max} \vert ^2 \right)$.
\end{proof}

\subsection{Proof of Lemma~\ref{lemma:agm}}\label{app:proof1}
\begin{proof}
The final step of Algorithm~\ref{alg:pbe_cfr} is to update the belief system $\mu$ given the average strategy $\bsigma^{*}$ using \textsc{UpdateBeliefs}. The plausibility order $\order$ according to $\bsigma^*$ is constructed using \textsc{ConstructOrderGivenProfile} so that the first and second requirements of AGM-consistency is satisfied. Specifically, if action $a$ is chosen at decision node $h \in H$ with positive probability $\bsigma^*(I(h))(a)$, then $h$ and $ha$ are equally plausible according to $\order$. 

Regarding the final requirement of AGM-consistency concerning $\mu$, we must separately consider those information sets that are specifically on the equilibrium path according to $r(\bsigma^*, \cdot)$ and those that are not.
In the first case, for any information set $I$ that is reached with probability $r(\bsigma^*, I) > 0$, \textsc{UpdateBeliefs} sets the belief $\mu^*(h \given I)$ according to $r$ using Bayes' rule in lines 19-21 (Algorithm~\ref{alg:update_beliefs}). All nodes $h \in I$ where $r(\bsigma^*, h) > 0$ therefore have positive belief $\mu^*(h \given I)$; let $B_+(I)$ denote these nodes. All nodes where $r(\bsigma^*, h) = 0$ have belief $\mu^*(h \given I) = 0$; let $B_0(I)$ denote these nodes. This means that at some preceding information set $I'$, $\bsigma^*$ specified an action $a$ (or actions) with positive probability, meaning $\order$ classified the nodes following $a$ (in $B_+(I)$) as more plausible than the nodes succeeding a different action (in $B_0(I)$) selected with zero probability. It also means that the nodes succeeding $a$ are all equally plausible. The transitive property of $\order$ tells us that if $h_1 \prec h_2$ and $h_1 \sim h_3$, then $h_3 \prec h_2$. This holds for all nodes succeeding $I'$, meaning that for $h_1 \in B_+(I)$ and $h_2 \in B_0(I)$, $h_1 \prec h_2$, and $h_1 \sim h'$, for all $h' \in B_+(I)$. The final requirement of AGM-consistency is therefore satisfied by $(\bsigma^*, \mu^*)$ for all information sets where $r(\bsigma^*, I) > 0$.

In the second case of an information set $I$ where $r(\bsigma^*, I) = 0$, we use the plausibility order $\order$ constructed around $\bsigma^*$ in order to extract the pairwise plausibility between any two nodes in $I$. For example, at an information set $I'$ preceding $I$, for every node $h \in I'$, suppose that $\bsigma^*$ specifies action $a$ rather than action $b$, yet the overall reach probability of $I$ given $\bsigma^*$ is 0. Suppose also that $\mathit{ha}, \mathit{hb} \in I$. $\order$ given $\bsigma^*$ indicates that $\mathit{ha} \prec \mathit{hb}$ because $b$ is chosen with zero probability. Therefore, these two nodes must be treated differently when updating $\mu^*$. \textsc{UpdateBeliefs} groups those nodes in $I$ that are all equally plausible into a ``plausibility clique" $V$, and $\mu^*(\cdot \given I)$ is uniformly distributed over all nodes in $V$. For all nodes $h \in I \setminus V$, $\mu(h \given I)$ is set to 0, and $h' \prec h$ is added to $\order$ for all $h' \in V$. This satisfies the final requirement of AGM-consistency for information sets off the equilibrium path, meaning the plausibility order $\order$ rationalizes $(\bsigma^*, \mu^*)$ and $(\bsigma^*, \mu^*)$ is AGM-consistent.

We revisit both cases to consider whether $\mu^*$ is Bayesian relative to $\order$. This is true in the first case because \textsc{UpdateBeliefs} sets the belief $\mu^*(h \given I)$ according to Bayes' rule using reach probabilities for every information set $I$ that is reachable by $\bsigma^*$. In the second case, at every information set $I$, its nodes are divided into two sets to demarcate plausibility. The most plausible nodes of $I$ yield a distribution where each is assigned positive $\mu^*$ while the least plausibile nodes have an assigned belief of 0. As stipulated by \citet{agm11}, this is sufficient. Therefore, $\mu^*$ is Bayesian relative to $\order$, meaning the second and third conditions for PBE are now met (Definition~\ref{def:pbe_agm11}).
\end{proof}

\subsection{Proof of Lemma~\ref{lemma:imm_bel_reg}}\label{app:proof2}
\begin{proof}
We expand the immediate believed regret of not playing action $a \in A_j(I)$ after running the algorithm for $T$ timesteps:
\begin{align*}
    R^T_{j, \mathit{imm}}(I)(a) &= \frac{1}{T} \sum_{t=1}^T U^B_j\left(\mu^t, \left.\bsigma^t\right|_{I \rightarrow a} \given I\right) - U^B_j\left(\mu^t, \bsigma^t \given I\right) \\
    &= \frac{1}{T} \sum_{t=1}^T U^B_j\left(\mu^t, \left.\bsigma^t\right|_{I \rightarrow a} \given I\right) - \sum_{h \in I} \mu^t(h \given I)  U^E_j\left( \bsigma^t \given h\right) \\
    &= \frac{1}{T} \sum_{t=1}^T U^B_j\left(\mu^t, \left.\bsigma^t\right|_{I \rightarrow a} \given I\right) - \sum_{h \in I} \mu^t(h \given I)  \sum_{a' \in A_j(I)} \bsigma^t(I)(a') U^E_j\left(\left.\bsigma^t\right|_{I \rightarrow a'} \given \mathit{ha'}\right) \\
    &=  \frac{1}{T} \sum_{t=1}^T U^B_j\left(\mu^t, \left.\bsigma^t\right|_{I \rightarrow a} \given I\right) - \sum_{h \in I} \mu^t(h \given I) \cdot \left( \sum_{a' \in A_j(I)} \bsigma^t(I)(a') U^E_j\left(\left.\bsigma^t\right|_{I \rightarrow a'} \given \mathit{ha'}\right) \right) \\
    &= \frac{1}{T} \sum_{t=1}^T U^B_j\left(\mu^t, \left.\bsigma^t\right|_{I \rightarrow a} \given I\right) - \sum_{a' \in A_j(I)} \bsigma^t(I)(a') \cdot \left( \sum_{h \in I} \mu^t(h \given I) \cdot U^E_j\left(\left.\bsigma^t\right|_{I \rightarrow a'} \given \mathit{ha'}\right) \right) \\
    &= \frac{1}{T} \sum_{t=1}^T U^B_j \left( \mu^t, \left.\bsigma^t\right|_{I \rightarrow a} \given I \right) - \sum_{a' \in A_j(I)} \bsigma^t(a') \cdot U^B_j \left( \mu^t, \left.\bsigma^t\right|_{I \rightarrow a'} \given I \right)
\end{align*}

Recall that regret-matching is defined in a domain where there is a fixed set of actions $A$, a function $u^t: A \rightarrow \mathbb{R}$, and a distribution over the actions $p^t$ selected at each timestep $t$. The regret of not playing action $a$ until time $T$ is given by
\begin{equation*}
    R^t(a) = \frac{1}{T} \sum_{t=1}^T u^t(a) - \sum_{a' \in A} p^t(a')u^t(a').
\end{equation*}

It is not hard to see that our expansion of $R^T_{j, \mathit{imm}}(I)(a)$ follows this format, when the belief is treated as part of the utility $u^t$ for playing the action $a$. Furthermore, the strategy $\bsigma^{t+1}$ is updated in the same fashion as in Blackwell's algorithm, computing the cumulative immediate believed regret $R^{T, +}_{j, \mathit{imm}}(I) = \max(0, R^T_{j, \mathit{imm}}(I))$. It follows that for any $h \in I$ and any $a \in A(I)$,
\begin{equation*}
   U^E_j\left(\left.\bsigma^t\right|_{I \rightarrow a} \given \mathit{ha}\right) - \sum_{a' \in A(I)} U^E_j\left(\bsigma^t \given \mathit{ha}\right) \cdot  \bsigma^t(I)(a') \leq \Delta_{u, j}.
\end{equation*}

We know that $\mu^t(\cdot \given I)$ sums to 1 over all nodes in $I$. Suppose that node $h \in I$ maximizes the difference in the expected utilities given by the previous equation. It follows that the difference in the believed utilities would be maximized when $\mu^t(h \given I) = 1$ and $\mu^t$ is equal to 0 for the remaining nodes in $I$, meaning

\begin{align*}
\sum_{h \in I} \mu^t(h \given I) \cdot \left( U^E_j\left(\left.\bsigma^t\right|_{I \rightarrow a} \given \mathit{ha}\right) - \sum_{a' \in A(I)} U^E_j\left(\bsigma^t \given \mathit{ha}\right) \cdot  \bsigma^t(I)(a') \right) &\leq \Delta_{u, j} \\
    U^B_j \left( \mu^t, \left.\bsigma^t\right|_{I \rightarrow a} \given I \right) - \bsigma^t(a) \cdot U^B_j \left( \mu^t, \left.\bsigma^t\right|_{I \rightarrow a} \given I \right) &\leq \Delta_{u, j}.
\end{align*}

Thus, for all $I \in \cI_j$, $a \in A(I)$, the immediate believed regret $R^T_{j, \mathit{imm}}(I)$ must be less than $\frac{\Delta_{u, j} \vert A_j(I) \vert}{\sqrt{T}}$, which we denote by $\varepsilon$. By definition, the immediate believed regret $R^{*}_{j, \mathit{imm}}(I)$ that results from our average strategy $\bsigma^{*}$ and associated consistent belief $\mu^{*}$ can be bounded from above by $R^T_{j, \mathit{imm}}(I)$ and therefore by $\varepsilon$ for sufficiently large $T$. 

\begin{align*}
    R^{*}_{j, \mathit{imm}}(I) &\leq R^T_{j, \mathit{imm}}(I) \leq \frac{\Delta_{u, j} \vert A_j(I) \vert}{\sqrt{T}} \equiv \varepsilon \\
    \implies T &\leq \left( \frac{\Delta_{u, j} \vert A_j(I) \vert}{\varepsilon} \right)^2
\end{align*}
\end{proof}

\subsection{Proof of Lemma~\ref{lemma:localBR}}\label{ap:proof3}
\begin{proof}
Let $\mathsf{Succ}(I, j)$ be the set of information sets $I' \in \cI_j$ that immediately succeed information set $I$. Let also $\pi'_I \pi_{j \setminus I}$ be the pure strategy where player $j$ plays $\pi'(I)$ at information set $I$ and $\pi_j$ elsewhere. We also define the probability $r^{\bsigma, \mu}(h' \given I)$ of reaching node $h'$ given that information set $I \neq I(h')$ has been reached according to assessment $(\bsigma, \mu)$:
\begin{equation*}
    r^{\bsigma, \mu}(h' \given I) = \sum_{h \in I} r(\bsigma, h' \given h) \cdot \mu(h \given I).
\end{equation*}
It follows that
\begin{align*}
    r^{\bsigma, \mu}(I' \given I) &= \sum_{h' \in I'} r^{\bsigma, \mu}(h' \given I) \\
    &= \sum_{h' \in I'} \sum_{h \in I} r(\bsigma, h' \given h) \cdot \mu(h \given I) \\
    &= \sum_{h \in I} \mu(h \given I) \sum_{h' \in I'} r(\bsigma, h' \given h).
\end{align*}

It is fairly obvious that the $(\rightarrow)$ direction is true; if a strategy is a sequential best response, it must also be a local best response at every information set.
To prove the $(\leftarrow)$ direction, we use backward induction to show that $\pi'_j$ maximizes the believed utility $U^B_j\left(\mu^{*}, \pi_j \bsigma^{*}_{-j} \given I\right)$ for all $I \in \cI_j$. First, we can reasonably assume that $\pi'_j$ maximizes $U^B_j\left(\mu^{*}, \pi_j \bsigma^{*}_{-j} \given I'\right)$ for all $I' \in \mathsf{Succ}(I, j)$ if the RHS is true.

Consider another strategy $\pi''_j \in \Pi_j$. Player $j$ believes he can acquire utility $U^B_j\left(\mu^{*}, \pi''_j \bsigma^{*}_{-j} \given I\right)$ when playing this strategy in response to the given assessment. We can rewrite this quantity as a sum of the contributions of $I' \in \mathsf{Succ}(I, j)$ to the believed utility and the contributions of the remaining paths leading to terminal nodes (either directly or through other players' decision nodes/information sets) to the believed utility:
\begin{equation*}
\begin{split}
    U^B_j\left(\mu^{*}, \pi''_j \bsigma^{*}_{-j} \given I\right) = \sum\limits_{\substack{z \in Z(I) \setminus \\ Z(\mathsf{Succ}(I, j)}} r^{\mu^{*}, \pi''_j \bsigma^{*}_{-j}}(z \given I) \cdot u_j(z) + 
    \sum\limits_{\substack{z \in Z(I')\\ I' \in \mathsf{Succ}(I, j) }} r^{\mu^{*}, \pi''_j \bsigma^{*}_{-j}}(z \given I) \cdot u_j(z)
\end{split}
\end{equation*}

The full believed regret at information set $I$ assuming $\pi''_j$ is a sequential best response to $(\mu^{*}, \bsigma^{*}_{-j})$ is given by 
\begin{align*}
    R^*_{j, full}(I) &= \left( \max\limits_{\pi''_j \in \Pi_j} \sum\limits_{\substack{z \in Z(I) \setminus \\ Z(\mathsf{Succ}(I, j)}} r^{\mu^{*}, \pi''_j \bsigma^*_{-j}}(z \given I) \cdot u_j(z) +\sum\limits_{\substack{z \in Z(I')\\ I' \in \mathsf{Succ}(I, j) }} r^{\mu^*, \pi''_j \bsigma^*_{-j}}(z \given I) \cdot u_j(z) \right) \\
    &- U^B_j\left(\mu^*, \bsigma^* \given I\right) \\
\end{align*}

Earlier, we divided $U^B_j\left(\mu, \pi''_j \bsigma_{-j} \given I\right)$ according to the contributions of the successor information sets $\mathsf{Succ}(I, j)$ and the remaining terminal nodes whose utilities are affected by $j$'s choice of action $a \in A_j(I)$. Similarly, we divide the strategy $\pi''_j$ into two disjoint parts: choosing $a \in A_J(I)$ to maximize the immediate utility and choosing the remainder of $\pi''_j$ for the information sets of $\mathsf{Succ}(I, j)$, denoted $\pi''_j \in \Pi_j \setminus A_j(I)$

\begin{align*}
    R^*_{j, full}(I) &= \max_{a \in A_j(I)\ } \max\limits_{\substack{\pi''_j \in \Pi_j \setminus \\ A_j(I)}} \sum\limits_{\substack{z \in Z(I) \setminus \\Z(\mathsf{Succ}(I, j, a)}} r^{\mu^*, \pi''_j \bsigma^*_{-j}}(z \given I) u_j(z) +\sum\limits_{\substack{z \in Z(I')\\ I' \in \mathsf{Succ}(I, j, a) }} r^{\mu^*, \pi''_j \bsigma^*_{-j}}(z \given I) u_j(z)
    \\
    &- U^B(\mu^*, \bsigma^* \given I) \\
\end{align*}

Notice that since the domain of $\pi''_j$ is now restricted to the action spaces of the successor information sets, the reach probability of any node $z \in Z(I) \setminus Z(\mathsf{Succ}(I, j, a)$ is independent of $\pi''_j$ and depends only on $\mu^*$ and $\bsigma^*_{-j}$:

\begin{align*}
    R^*_{j, full}(I) &= \max_{a \in A_j(I) \ } \sum\limits_{\substack{z \in Z(I) \setminus \\Z(\mathsf{Succ}(I, j, a)}} r^{\mu^*, \bsigma_{-j}^*}(z \given I) u_j(z) 
+ \max\limits_{\substack{\pi''_j \in \Pi_j \setminus \\ A_j(I)}} \sum\limits_{\substack{z \in Z(I')\\ I' \in \mathsf{Succ}(I, j, a) }} r^{\mu^*, \pi''_j \bsigma^*_{-j}}(z \given I) u_j(z) \\
&- U^B_j\left(\mu^*, \bsigma^* \given I\right)
\end{align*}

We now expand $U^B_j\left(\mu^*, \bsigma^* \given I\right)$:

\begin{align*}
    R^*_{j, full}(I) &= \max_{a \in A_j(I)} \sum\limits_{\substack{z \in Z(I) \setminus \\Z(\mathsf{Succ}(I, j, a)}} r^{\mu^*, \bsigma_{-j}^*}(z \given I) u_j(z) 
+ \max\limits_{\substack{\pi''_j \in \Pi_j \setminus \\ A_j(I)}} \sum\limits_{\substack{z \in Z(I')\\ I' \in \mathsf{Succ}(I, j, a) }} r^{\mu^*, \pi''_j \bsigma^*_{-j}}(z \given I) u_j(z) \\
&- \left(\sum\limits_{\substack{z \in Z(I) \setminus \\Z(\mathsf{Succ}(I, j)}} r^{\mu^*, \bsigma^*}(z \given I)u_j(z) + \sum\limits_{\substack{z \in Z(I') \\ I' \in \mathsf{Succ}(I, j)}} r^{\mu^*, \bsigma^*}(z \given I) u_j(z)
\right) 
\end{align*}

Given that $(\mu^*, \bsigma^*)$ is consistent, it follows that for each player $j$, any pair of information sets $I, I' \in \cI_j$, where $I'$ succeeds $I$, and for any terminal node $z$ succeeding $I'$, and any strategy $\pi'_j \in \Pi_j$ \citep{hendon94}:
\begin{equation*}
    r^{\mu^*, \pi'_j \bsigma^*_{-j}}(z \given I) = r^{\mu^*, \pi'_j \bsigma^*_{-j}}(z \given I') \cdot r^{\mu^*, \pi'_j \bsigma^*_{-j}}(I' \given I).
\end{equation*}

We make this substitution to the reach probability associated with the terminal nodes following any information set succeeding $I$:
\begin{align*}
    R^*_{j, full}(I) &= \max_{a \in A_j(I)} \sum\limits_{\substack{z \in Z(I) \setminus \\Z(\mathsf{Succ}(I, j, a)}} r^{\mu^*, \bsigma_{-j}^*}(z \given I) u_j(z) 
+ \max\limits_{\substack{\pi''_j \in \Pi_j \setminus \\ A_j(I)}} \sum\limits_{\substack{z \in Z(I')\\ I' \in \mathsf{Succ}(I, j, a) }} r^{\mu^*, \pi''_j \bsigma^*_{-j}}(I' \given I) r^{\mu^*, \pi''_j \bsigma^*_{-j}}(z \given I') u_j(z) \\
&- \left(\sum\limits_{\substack{z \in Z(I) \setminus \\Z(\mathsf{Succ}(I, j)}} r^{\mu^*, \bsigma^*}(z \given I)u_j(z) + \sum\limits_{\substack{z \in Z(I') \\ I' \in \mathsf{Succ}(I, j)}} r^{\mu^*, \bsigma^*}(z \given I') r^{\mu^*, \bsigma^*}(I' \given I) u_j(z) \right) 
\end{align*}

By definition, this is actually equal to

\begin{align*}
    R^*_{j, full}(I) &= \max_{a \in A_j(I)} \sum\limits_{\substack{z \in Z(I) \setminus \\ Z(\mathsf{Succ}(I, j, a)}} r^{\mu^*, \bsigma_{-j}^*}(z \given I) u_j(z) + \max\limits_{\substack{\pi''_j \in \Pi_j \setminus \\ A_j(I)}} \sum\limits_{I' \in \mathsf{Succ}(I, j, a)}  r^{\mu^*, \pi''_j \bsigma^*_{-j}}(I' \given I) U^B_j\left(\mu^*, \pi''_j \bsigma^*_{-j} \given I'\right) \\
    &- \left(\sum\limits_{\substack{z \in Z(I) \setminus \\Z(\mathsf{Succ}(I, j)}} r^{\mu^*, \bsigma^*}(z \given I)u_j(z) + \sum\limits_{I' \in \mathsf{Succ}(I, j)} r^{\mu^*, \bsigma^*}(I' \given I) U^B_j\left(\mu^*, \bsigma^* \given I') \right) \right).
\end{align*}

Because the induction hypothesis tells us that $\pi'_j$ maximizes $U^B_j\left(\mu^*, \pi''_j \bsigma^*_{-j} \given I'\right)$ for all $I'$, 

\[ U^B_j\left(\mu^*, \pi''_j \bsigma^*_{-j} \given I'\right) \leq U^B_j\left(\mu^*, \pi''_I \pi'_{j \setminus I} \bsigma^*_{-j} \given I\right),\] 

where $\pi''_I$ specifies the action $a \in A_j(I)$ that maximizes immediate believed regret. We can now bound the full believed regret:
\begin{align*}
    R^*_{j, full}(I) &\leq \max_{a \in A_j(I)} \sum\limits_{\substack{z \in Z(I) \setminus \\ Z(\mathsf{Succ}(I, j, a)}} r^{\mu^*, \bsigma_{-j}^*}(z \given I) u_j(z) + \sum\limits_{I' \in \mathsf{Succ}(I, j, a)}  r^{\mu^*, \pi''_I \pi'_{j \setminus I} \bsigma^*_{-j}}(I' \given I) U^B_j\left(\mu^*, \pi''_I \pi'_{j \setminus I} \bsigma^*_{-j} \given I'\right) \\
    &- \left(\sum\limits_{\substack{z \in Z(I) \setminus \\Z(\mathsf{Succ}(I, j)}} r^{\mu^*, \bsigma^*}(z \given I)u_j(z) + \sum\limits_{I' \in \mathsf{Succ}(I, j)} r^{\mu^*, \bsigma^*}(I' \given I) U^B_j\left(\mu^*, \bsigma^* \given I'\right)
    \right).
\end{align*}

We now rearrange the regret equation to group terms by $I' \in \mathsf{Succ}(I, j)$ and the remainder of the game after $I$. The first term is the \term{immediate believed regret} $R^*_{j, \mathit{imm}}(I)$ at $I$; the remaining regret for $I' \in \mathsf{Succ}(I, j, a)$ assumes $a$ is played at $I$ to maximize $R^*_{j, \mathit{imm}}(I)$.

\begin{align*}
    R^*_{j, full}(I) &\leq \max_{a \in A_j(I)} \sum\limits_{\substack{z \in Z(I) \setminus \\ Z(\mathsf{Succ}(I, j, a)}} r^{\mu^*, \bsigma_{-j}^*}(z \given I) u_j(z) - \sum\limits_{\substack{z \in Z(I) \setminus \\Z(\mathsf{Succ}(I, j)}} r^{\mu^*, \bsigma^*}(z \given I)u_j(z) \\
    &+ \sum\limits_{I' \in \mathsf{Succ}(I, j, a)}  r^{\mu^*, \pi''_I \pi'_{j \setminus I} \bsigma^*_{-j}}(I' \given I) U^B_j\left(\mu^*, \pi''_I \pi'_{j \setminus I} \bsigma^*_{-j} \given I'\right) - \sum\limits_{I' \in \mathsf{Succ}(I, j)} r^{\mu^*, \bsigma^*}(I' \given I) U^B_j\left(\mu^*, \bsigma^* \given I'\right) \\
\end{align*}

$\pi''_I$ must be assigned to $a \in A_j(I)$ that maximizes the believed immediate regret $R^T_{j, \mathit{imm}}(I)$ with probability 1. But this also means that $a$ is the local best response at $I$ which we know is specified by $\pi'_j(I)$, so $\pi''_I(I) \equiv \pi'_j(I)$, so it follows that

\begin{align*}
    R^*_{j, full}(I) &\leq \max_{a \in A_j(I)} \sum\limits_{\substack{z \in Z(I) \setminus \\ Z(\mathsf{Succ}(I, j, a)}} r^{\mu^*, \bsigma_{-j}^*}(z \given I) u_j(z) - \sum\limits_{\substack{z \in Z(I) \setminus \\Z(\mathsf{Succ}(I, j)}} r^{\mu^*, \bsigma^*}(z \given I)u_j(z) \\
    &+ \sum\limits_{I' \in \mathsf{Succ}(I, j, a)}  r^{\mu^*, \pi'_j \bsigma^*_{-j}}(I' \given I) U^B_j\left(\mu^*, \pi'_j \bsigma^*_{-j} \given I'\right) - \sum\limits_{I' \in \mathsf{Succ}(I, j)} r^{\mu^*, \bsigma^*}(I' \given I) U^B_j\left(\mu^*, \bsigma^* \given I'\right) \\
    &\leq U^B_j\left(\mu^*, \pi'_j \bsigma^*_{-j} \given I\right) - U^B_j\left(\mu^*, \bsigma^* \given I\right) \\
\end{align*}

Since we are forced to conclude that $U^B_j\left(\mu^*, \pi''_j \bsigma^*_{-j} \given I\right) \leq U^B_j\left(\mu^*, \pi'_j \bsigma^*_{-j} \given I\right)$ for any strategy $\pi''_j$, the believed regret induced by the joint strategy $\pi'_j$ of all local best responses succeeding $I$ must be greater than that of $\pi''_j$. It follows that $\pi'_j$ satisfies sequential rationality at information set $I$, and by induction the entire game tree.
\end{proof}

\subsection{Proof of Lemma~\ref{lem:imm_believed}}\label{app:proof_imm_believed}
\begin{proof}
We know that for all $I \in \cI_j$, the immediate believed regret $R^T_{j, \mathit{imm}}(I)$ (and therefore $R^{*}_{j, \mathit{imm}}(I)$) is bounded from above by $\frac{\Delta_{u, j} \vert A_j(I) \vert}{\sqrt{T}}$. For large enough $T$, this means that local sequential rationality is satisfied for $I$. Due to the one-shot deviation principle, we know that given a consistent assessment $(\bsigma^{*}, \mu^{*})$, if the strategy $\pi'_j$ is a local best response minimizing immediate regret to be at most $\frac{\Delta_{u, j} \vert A_j(I) \vert}{\sqrt{T}}$ at every $I$, then it must also be a sequential best response.
\end{proof}

\section{Games based on Goofspiel}

\subsection{\gengoof: Abstracted General-Sum Version of Goofspiel}\label{app:games_gengoof}

We define a new parameterized class of two-player general-sum games of imperfect information adapted from the game Goofspiel: a multiplayer symmetric zero-sum card game invented by Merrill Flood at Princeton University in the 1930s \citep{flood_goof_web30s}. In Goofspiel, a public card is drawn from the deck of point cards, and the players each bid on the card by playing a single card from their own hand of point cards simultaneously. Whoever bids highest wins the point card and earns the number of points associated with that card. Gameplay continues in a series of bidding rounds until the deck of point cards is gone. The player with the most points accumulated at the end wins.

Goofspiel has been a common game of choice for mathematical study and for evaluating the performance of AI and game-theoretic algorithms on multi-round multiplayer games that require considerable strategic thinking \citep{fixx_goof72, parlett_goof00}. For instance, \citet{ross_goof71} studied the case of Goofspiel for two players where one player randomly played his bidding cards to determine the optimal strategy for the other player, who bid strategically, and found that the best strategy was to match the public card. \citet{rhoads_goof12} computed a mixed Nash equilibrium for the game formulated by \citet{ross_goof71} by organizing the game into subgames represented as matrices and solving each matrix using linear and dynamic programming.

The parameterized, modified version \gengoofK{K} introduced by \citet{konicki25} differs from Goofspiel in the following ways:
\begin{itemize}
    \item The stochastic events and player actions of \gengoof are fundamentally abstract; the goal is to retain the basic structure and information flow of Goofspiel without necessarily having players bid on a public card.
    \item \gengoof also proceeds in rounds, but the number of rounds $K$ is a parameter of interest rather than always being set to 13 in the case of Goofspiel (the point deck contains all cards of the same suit). This allows the size and complexity of the game to be customized.
    \item \gengoof is a general-sum game while Goofspiel is a zero-sum game.
\end{itemize}

\begin{figure}[htbp]
\centering
\begin{subfigure}[t]{0.95\textwidth}
\centering
\begin{tikzpicture}[thick,
    level 1/.style = {level distance = 15mm, sibling distance = 22mm},
    level 2/.style = {level distance = 15mm, sibling distance = 26mm},
    level 3/.style = {level distance = 15mm, sibling distance = 32mm},
    level 4/.style = {level distance = 18mm, sibling distance = 12mm},
    engine/.style = {inner sep = 1pt, above}]
    \node [draw, black, fill={gg_chance}, regular polygon, regular polygon sides=3, rotate=180, inner sep=0.10cm] {}
    [black, ->]
    child { node [draw, black, fill={gg_player1}, circle] (1) {} 
        child {node [draw, black, fill={gg_player2}, diamond] (2) {} 
            child {node [draw, black, fill={gg_chance}, regular polygon, regular polygon sides=3, rotate=180, inner sep=0.10cm] (3) {}
                child {node [draw, black, fill={gg_player1_4}, circle] (4) {}
                edge from parent node[engine, left] {$B$}}
                child {node [draw, black, fill={gg_player1_5}, circle] (5) {}
                edge from parent node[engine, left] {$C$}}
                child {node [draw, black, fill={gg_player1_6}, circle] (6) {}
                edge from parent node[engine, right] {$D$}}
            edge from parent node[engine, sloped] {$a^1_2$}}
            child {node [draw, black, fill={gg_chance}, regular polygon, regular polygon sides=3, rotate=180, inner sep=0.10cm] (7) {}
                child {node [draw, black, fill={gg_player1_7}, circle] (8) {}
                edge from parent node[engine, left] {$B$}}
                child {node [draw, black, fill={gg_player1_8}, circle] (9) {}
                edge from parent node[engine, left] {$C$}}
                child {node [draw, black, fill={gg_player1_9}, circle] (10) {}
                edge from parent node[engine, right] {$D$}}
        edge from parent node[engine, sloped] {$a^4_2$}}
        edge from parent node[engine, sloped] {$a^1_1$}}
        child {node [draw, black, fill={gg_player2}, diamond] (11) {} 
        edge from parent node[engine, sloped] {$a^4_1$}}
    edge from parent node[engine, sloped] {$A$}}
    child { node [draw, black, fill={gg_player1_1}, circle] (12) {} 
    edge from parent node[engine, left] {$B$}}
    child { node [draw, black, fill={gg_player1_2}, circle] (13) {} 
    edge from parent node[engine, right] {$C$}}
    child { node [draw, black, fill={gg_player1_3}, circle] (14) {} 
        child {node [draw, black, fill={gg_player2_1}, diamond] (17) {} 
        edge from parent node[engine, sloped] {$a^1_1$}}
        child {node [draw, black, fill={gg_player2_1}, diamond] (26) {} 
            child {node [draw, black, fill={gg_chance}, regular polygon, regular polygon sides=3, rotate=180, inner sep=0.10cm] (18) {}
        child {node [draw, black, fill={gg_player1_10}, circle] (19) {}
        edge from parent node[engine, left] {$A$}}
        child {node [draw, black, fill={gg_player1_11}, circle] (20) {}
        edge from parent node[engine, left] {$B$}}
        child {node [draw, black, fill={gg_player1_12}, circle] (21) {}
        edge from parent node[engine, right] {$C$}}
    edge from parent node[engine, sloped] {$a^1_2$}}
    child {node [draw, black, fill={gg_chance}, regular polygon, regular polygon sides=3, rotate=180, inner sep=0.10cm] (22) {}
        child {node [draw, black, fill={gg_player1_13}, circle] (23) {}
        edge from parent node[engine, left] {$A$}}
        child {node [draw, black, fill={gg_player1_14}, circle] (24) {}
        edge from parent node[engine, left] {$B$}}
        child {node [draw, black, fill={gg_player1_15}, circle] (25) {}
        edge from parent node[engine, right] {$C$}}
edge from parent node[engine, sloped] {$a^4_2$}}
        edge from parent node[engine, sloped] {$a^4_1$}}
    edge from parent node[engine, sloped] {$D$}
    };
    \node[below=1cm of 12] (15) {};
    \node[below=1cm of 13] (16) {};
    \node[below=1cm of 11] (27) {};
    \node[below=1cm of 17] (28) {};
    \node[below=1cm of 4] (29) {};
    \node[below=1cm of 5] (30) {};
    \node[below=1cm of 6] (31) {};
    \node[below=1cm of 8] (32) {};
    \node[below=1cm of 9] (33) {};
    \node[below=1cm of 10] (34) {};
    \node[below=1cm of 19] (35) {};
    \node[below=1cm of 20] (36) {};
    \node[below=1cm of 21] (37) {};
    \node[below=1cm of 23] (38) {};
    \node[below=1cm of 24] (39) {};
    \node[below=1cm of 25] (40) {};

    \path (12) -- node[auto=false]{{\LARGE \vdots}} (15);
    \path (13) -- node[auto=false]{{\LARGE \vdots}} (16);
    \path (2) -- node[auto=false]{{\LARGE \ldots}} (11);
    \path (3) -- node[auto=false]{{\LARGE \ldots}} (7);
    \path (11) -- node[auto=false]{{\LARGE \vdots}} (27);
    \path (18) -- node[auto=false]{{\LARGE \ldots}} (22);
    \path (17) -- node[auto=false]{{\LARGE \vdots}} (28);
    \path (4) -- node[auto=false]{{\LARGE \vdots}} (29);
    \path (5) -- node[auto=false]{{\LARGE \vdots}} (30);
    \path (6) -- node[auto=false]{{\LARGE \vdots}} (31);
    \path (8) -- node[auto=false]{{\LARGE \vdots}} (32);
    \path (9) -- node[auto=false]{{\LARGE \vdots}} (33);
    \path (10) -- node[auto=false]{{\LARGE \vdots}} (34);
    \path (19) -- node[auto=false]{{\LARGE \vdots}} (35);
    \path (20) -- node[auto=false]{{\LARGE \vdots}} (36);
    \path (21) -- node[auto=false]{{\LARGE \vdots}} (37);
    \path (23) -- node[auto=false]{{\LARGE \vdots}} (38);
    \path (24) -- node[auto=false]{{\LARGE \vdots}} (39);
    \path (25) -- node[auto=false]{{\LARGE \vdots}} (40);
    
\end{tikzpicture}
\end{subfigure}
\begin{subfigure}[t]{0.9\textwidth}
\centering
\fbox{\begin{tabular}{rr}
\small{\textcolor{gg_chance}{\TriangleDown}} & \footnotesize Chance Nodes \\
\small{\textcolor{gg_player1}{\CircleSolid} \textcolor{gg_player1_1}{\CircleSolid} \textcolor{gg_player1_2}{\CircleSolid} \textcolor{gg_player1_3}{\CircleSolid} \textcolor{gg_player1_4}{\CircleSolid} \textcolor{gg_player1_5}{\CircleSolid} \textcolor{gg_player1_6}{\CircleSolid} \textcolor{gg_player1_7}{\CircleSolid}} \\ 
\small{\textcolor{gg_player1_8}{\CircleSolid} \textcolor{gg_player1_9}{\CircleSolid} \textcolor{gg_player1_10}{\CircleSolid} \textcolor{gg_player1_11}{\CircleSolid} \textcolor{gg_player1_12}{\CircleSolid} \textcolor{gg_player1_13}{\CircleSolid} \textcolor{gg_player1_14}{\CircleSolid} \textcolor{gg_player1_15}{\CircleSolid}} & \footnotesize Player 1 Infosets \\
\small{\textcolor{gg_player2}{\DiamondSolid} \textcolor{gg_player2_1}{\DiamondSolid}} & \footnotesize Player 2 Infosets
\end{tabular}}
    \end{subfigure}
    \caption{EFG representation of the first round of \gengoof}
    \label{fig:gengoof_chap3}
\end{figure}
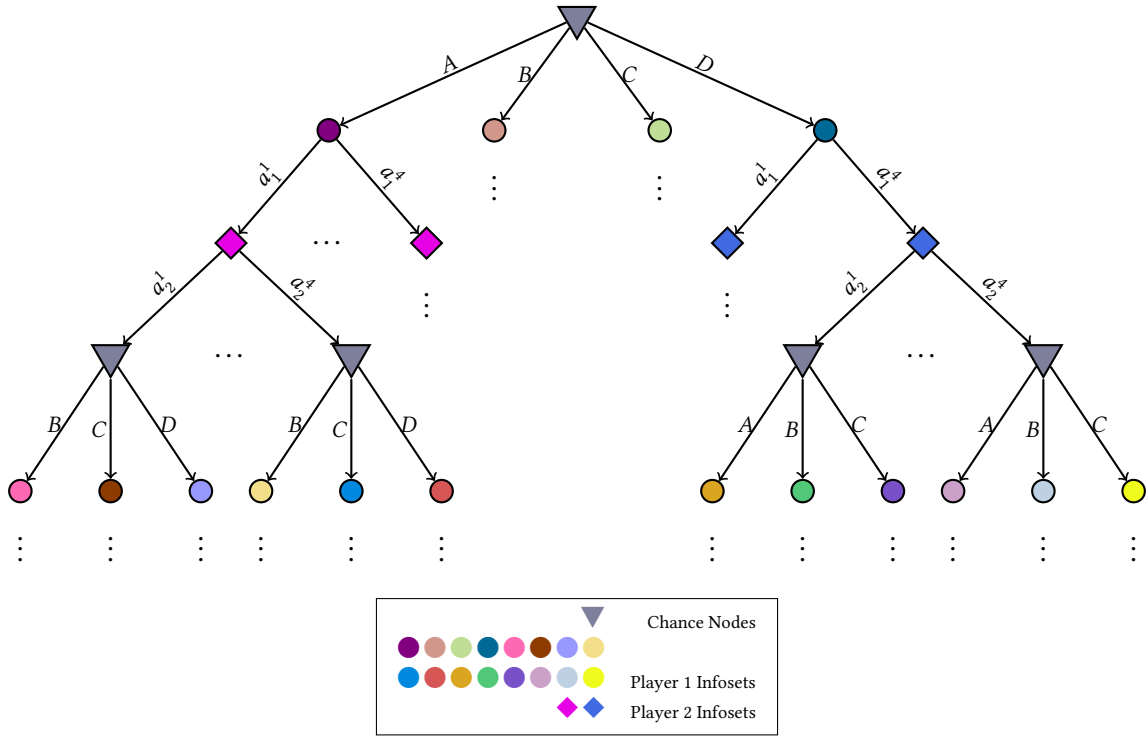

A single round of \gengoof is visualized in Figure~\ref{fig:gengoof_chap3}, where the number of stochastic outcomes at the start of the game is $K=4$, and the total number of game rounds is $K-1$. This particular version of \gengoof is denoted by \gengoofK4. Gameplay proceeds as follows. First, a stochastic event with $K$ outcomes occurs at the game root; the event is denoted
\begin{equation*}
     e_1 \in X(\emptyset)
\end{equation*}
where $X(\emptyset)$ is the set of the first $K$ letters of an alphabet. In our \gengoofK4 example in  Figure~\ref{fig:gengoof_chap3},  $X(\emptyset)=\{A,B,C,D\}$. 
Player~1 observes outcome $e_1$ and chooses one of $K$ actions from
\begin{equation*}
    \Pi_1(I(e_1)) = \{ a^k_1 \}_{k=1}^K.
\end{equation*}
Player~2 then observes $e_1$ but not player~1's action, and also chooses one of $K$ actions from
\begin{equation*}
    \Pi_2(I(e_1 a^k_1)) = \{ a^k_2 \}_{k=1}^K.
\end{equation*}
Then, a second stochastic event $e_2$ with $K - 1$ possible outcomes (excluding $e_1$) occurs:
\begin{equation*}
    X(e_1 a^k_1 a^k_2) = X(\emptyset)  \setminus \{ e_1 \}; \quad  e_2 \in X(e_1 a^k_1 a^k_2).
\end{equation*}
Player~1 then observes the history of all actions up to and including $e_2$ before choosing one of $K$ actions, followed by player~2 (who observes all but player~1's second chosen action). This process repeats until round $K-1$ where there are only 2 possible outcomes for the final stochastic event that occurs, followed by player~1 and player~2 each taking a turn as before. 

For each instance of \gengoof, we sample a categorical probability distribution, denoted by $\Prob(\cdot \given \emptyset)$, uniformly at random from the $(K-1)$-simplex for the round-$1$ stochastic event; for $k\in\{2,3,\dots,K-1\}$, we renormalize the distribution over the residual support after eliminating the outcome realized in round $(k-1)$. For example, the probability distribution of the round-$2$ stochastic event given that $e_1$ occurred in round $1$ is
\begin{equation*}
    P\left(e_2 \given e_1 a^k_1 a^k_2 \right) = \frac{P(e_2 \given \emptyset)}{\sum\limits_{e' \in X(\emptyset) \setminus \{ e_1 \}} P(e' \given \emptyset)} \quad \forall e_2 \in X(\emptyset) \setminus \{ e_1 \}.
\end{equation*}

For each possible combination of the stochastic outcome and the two players' action choices in each round of any game instance, we choose a reward for each player uniformly at random from $[0, u_{\max}]$ for a positive real number $u_{\max}$; we set the utility for each player on game termination equal to the sum of the player's rewards over all $K - 1$ rounds in that history. Thus, for every leaf $z \in Z$ and player $j\in\{1,2\}$, $u_j(z) \sim \Udist([0,u_{\max}(K-1)])$ where $\Udist(\cS)$ denotes the uniform distribution over the set $\cS$.

For our experiments in Section~\ref{sec:pbe_exp}, we used $K\in\{4,5\}$ and $u_{\max}=10$; all utility and probability parameters of the underlying game were hidden from the game theorist applying TE-PSRO to \gengoof.

\subsection{\modgg: Modification of \gengoof with Private Chance Events}\label{app:games_abs_private}


We obtain \modggK{K} by modifying the flow of information in \gengoofK{K} so that neither player $1$ nor player $2$ observes the chance outcome in each round before making their respective moves, but player~$2$ chooses its own action after observing the action player~$1$ who still cannot observe player~$2$'s action. The probability distributions, action spaces, and utility structure are unchanged, and the realized rewards in each round are still publicly observable at the end of the respective rounds. In our experiments in Section~\ref{sec:pbe_exp}, we used the same parameters (hidden from the game theorist) as those specified for \gengoof above. A single round of \modggK4 is visualized in Figure~\ref{fig:private_gengoof_chap3}.

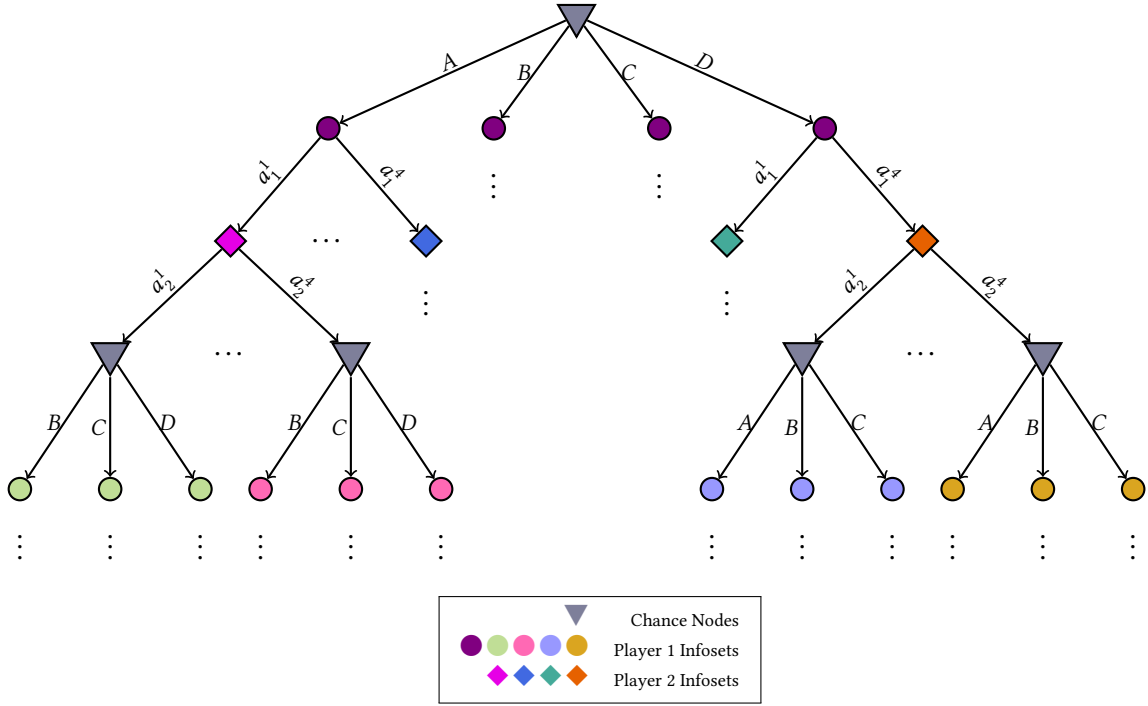
\begin{figure}[htbp]
\centering
\begin{subfigure}[t]{0.89\textwidth}
\hspace{-20pt}
\centering
\begin{tikzpicture}[thick,
    level 1/.style = {level distance = 15mm, sibling distance = 22mm},
    level 2/.style = {level distance = 15mm, sibling distance = 26mm},
    level 3/.style = {level distance = 15mm, sibling distance = 32mm},
    level 4/.style = {level distance = 18mm, sibling distance = 12mm},
    engine/.style = {inner sep = 1pt, above}]
    \node [draw, black, fill={gg_chance}, regular polygon, regular polygon sides=3, rotate=180, inner sep=0.10cm] {}
    [black, ->]
    child { node [draw, black, fill={gg_player1}, circle] (1) {} 
        child {node [draw, black, fill={gg_player2}, diamond] (2) {} 
            child {node [draw, black, fill={gg_chance}, regular polygon, regular polygon sides=3, rotate=180, inner sep=0.10cm] (3) {}
                child {node [draw, black, fill={gg_player1_2}, circle] (4) {}
                edge from parent node[engine, left] {$B$}}
                child {node [draw, black, fill={gg_player1_2}, circle] (5) {}
                edge from parent node[engine, left] {$C$}}
                child {node [draw, black, fill={gg_player1_2}, circle] (6) {}
                edge from parent node[engine, right] {$D$}}
            edge from parent node[engine, sloped] {$a^1_2$}}
            child {node [draw, black, fill={gg_chance}, regular polygon, regular polygon sides=3, rotate=180, inner sep=0.10cm] (7) {}
                child {node [draw, black, fill={gg_player1_4}, circle] (8) {}
                edge from parent node[engine, left] {$B$}}
                child {node [draw, black, fill={gg_player1_4}, circle] (9) {}
                edge from parent node[engine, left] {$C$}}
                child {node [draw, black, fill={gg_player1_4}, circle] (10) {}
                edge from parent node[engine, right] {$D$}}
        edge from parent node[engine, sloped] {$a^4_2$}}
        edge from parent node[engine, sloped] {$a^1_1$}}
        child {node [draw, black, fill={gg_player2_1}, diamond] (11) {} 
        edge from parent node[engine, sloped] {$a^4_1$}}
    edge from parent node[engine, sloped] {$A$}}
    child { node [draw, black, fill={gg_player1}, circle] (12) {} 
    edge from parent node[engine, left] {$B$}}
    child { node [draw, black, fill={gg_player1}, circle] (13) {} 
    edge from parent node[engine, right] {$C$}}
    child { node [draw, black, fill={gg_player1}, circle] (14) {} 
        child {node [draw, black, fill={gg_player2_2}, diamond] (17) {} 
        edge from parent node[engine, sloped] {$a^1_1$}}
        child {node [draw, black, fill={gg_player2_3}, diamond] (26) {} 
            child {node [draw, black, fill={gg_chance}, regular polygon, regular polygon sides=3, rotate=180, inner sep=0.10cm] (18) {}
        child {node [draw, black, fill={gg_player1_6}, circle] (19) {}
        edge from parent node[engine, left] {$A$}}
        child {node [draw, black, fill={gg_player1_6}, circle] (20) {}
        edge from parent node[engine, left] {$B$}}
        child {node [draw, black, fill={gg_player1_6}, circle] (21) {}
        edge from parent node[engine, right] {$C$}}
    edge from parent node[engine, sloped] {$a^1_2$}}
    child {node [draw, black, fill={gg_chance}, regular polygon, regular polygon sides=3, rotate=180, inner sep=0.10cm] (22) {}
        child {node [draw, black, fill={gg_player1_10}, circle] (23) {}
        edge from parent node[engine, left] {$A$}}
        child {node [draw, black, fill={gg_player1_10}, circle] (24) {}
        edge from parent node[engine, left] {$B$}}
        child {node [draw, black, fill={gg_player1_10}, circle] (25) {}
        edge from parent node[engine, right] {$C$}}
edge from parent node[engine, sloped] {$a^4_2$}}
        edge from parent node[engine, sloped] {$a^4_1$}}
    edge from parent node[engine, sloped] {$D$}
    };
    \node[below=1cm of 12] (15) {};
    \node[below=1cm of 13] (16) {};
    \node[below=1cm of 11] (27) {};
    \node[below=1cm of 17] (28) {};
    \node[below=1cm of 4] (29) {};
    \node[below=1cm of 5] (30) {};
    \node[below=1cm of 6] (31) {};
    \node[below=1cm of 8] (32) {};
    \node[below=1cm of 9] (33) {};
    \node[below=1cm of 10] (34) {};
    \node[below=1cm of 19] (35) {};
    \node[below=1cm of 20] (36) {};
    \node[below=1cm of 21] (37) {};
    \node[below=1cm of 23] (38) {};
    \node[below=1cm of 24] (39) {};
    \node[below=1cm of 25] (40) {};

    \path (12) -- node[auto=false]{{\LARGE \vdots}} (15);
    \path (13) -- node[auto=false]{{\LARGE \vdots}} (16);
    \path (2) -- node[auto=false]{{\LARGE \ldots}} (11);
    \path (3) -- node[auto=false]{{\LARGE \ldots}} (7);
    \path (11) -- node[auto=false]{{\LARGE \vdots}} (27);
    \path (18) -- node[auto=false]{{\LARGE \ldots}} (22);
    \path (17) -- node[auto=false]{{\LARGE \vdots}} (28);
    \path (4) -- node[auto=false]{{\LARGE \vdots}} (29);
    \path (5) -- node[auto=false]{{\LARGE \vdots}} (30);
    \path (6) -- node[auto=false]{{\LARGE \vdots}} (31);
    \path (8) -- node[auto=false]{{\LARGE \vdots}} (32);
    \path (9) -- node[auto=false]{{\LARGE \vdots}} (33);
    \path (10) -- node[auto=false]{{\LARGE \vdots}} (34);
    \path (19) -- node[auto=false]{{\LARGE \vdots}} (35);
    \path (20) -- node[auto=false]{{\LARGE \vdots}} (36);
    \path (21) -- node[auto=false]{{\LARGE \vdots}} (37);
    \path (23) -- node[auto=false]{{\LARGE \vdots}} (38);
    \path (24) -- node[auto=false]{{\LARGE \vdots}} (39);
    \path (25) -- node[auto=false]{{\LARGE \vdots}} (40);
    
\end{tikzpicture}
\end{subfigure}
\begin{subfigure}[t]{0.9\textwidth}
\centering
\fbox{\begin{tabular}{rr}
\small{\textcolor{gg_chance}{\TriangleDown}} & \footnotesize Chance Nodes \\
\small{\textcolor{gg_player1}{\CircleSolid} \textcolor{gg_player1_2}{\CircleSolid} \textcolor{gg_player1_4}{\CircleSolid} \textcolor{gg_player1_6}{\CircleSolid} \textcolor{gg_player1_10}{\CircleSolid}} & \footnotesize Player 1 Infosets \\
\small{\textcolor{gg_player2}{\DiamondSolid} \textcolor{gg_player2_1}{\DiamondSolid} \textcolor{gg_player2_2}{\DiamondSolid} \textcolor{gg_player2_3}{\DiamondSolid}} & \footnotesize Player 2 Infosets
\end{tabular}}
    \end{subfigure}
    \caption{EFG representation of the first round of \modggK4}
    \label{fig:private_gengoof_chap3}
\end{figure}

\section{\dondG: Sequential Bargaining Game with Outside Offers}\label{app:games_bargain}

We adapted this game from a well-known alternating-offer bargaining game of incomplete information that has been used in a variety of AI studies, even featuring the use of large-language models \citep{devault_dond15, lewis_dond17, strouse_dond21, li25lmmghlbw}. 
The adapted version strips away all natural-language elements from the original, but adds an external fallback option to augment the informational and strategic complexity. 

In \dondG, two players negotiate the split of a public pool $\numItems$ discrete items of $\tau$ distinct types between themselves.
We represent the item pool by a vector $\mathbf{p}$ where the $i^\mathrm{th}$ entry $p_i$ is the number of items of type $i \in \{1, \dotsc, \tau \} \equiv [\tau]$:
\begin{equation*}
    \sum_{i=1}^{\tau} p_1 = \numItems.
\end{equation*}

Each player~$j \in \{1,2\}$ has a private valuation over the items represented by a vector $\mathbf{v}_j$ of non-negative integers such that the $i^\mathrm{th}$ entry $v_{j,i}$ is player~$j$'s value for one item of type~$i$. 
In each game instance, $(\mathbf{v}_1, \mathbf{v}_2)$ is sampled uniformly at random from the collection $\cV$ of all vector pairs satisfying the following three constraints:
\begin{enumerate}
\item For each player, the total value of all items is the same constant:
\begin{equation*}
    \forall  j\in\{1,2\}. \mathbf{v}_j \cdot \mathbf{p} = \bar{V}
\end{equation*}
\item Each item type has a nonzero value to at least one player: 
\begin{equation*}
    \forall i \in [\tau]. v_{1, i} + v_{2, i} > 0 
\end{equation*}
\item At least one item type has a nonzero value to each player: 
\begin{align*}
 \exists i \in [\tau]. v_{1, i} v_{2, i} > 0
\end{align*} 
\end{enumerate}
An additional feature of \dondG is that each player~$j$ access to a private \term{outside offer} in the form of a vector of items $\mathbf{o}_j$ of the same $\tau$ types as in the above pool. This offer represents the fallback payoff that each player obtains if negotiation "fails," i.e., if one player walks away or if no deal is reached by the end of the game (see below).
This offer is drawn from a distribution $P_j(\cdot)$ for each player $j \in\{1,2\}$ at the start of each game instance. 
Moreover, during negotiation, player~$j$ may choose to reveal coarsened information about its outside offer to the other player in the form of a binary signal which is $L$ (resp.~$H$) if the value of the offer $\mathbf{o}_j \cdot \mathbf{v}_j$ is at most (resp.\ greater than) a fixed threshold~$\thresh$ where $1 < \thresh < \bar{V}$ where $\bar{V}$ is the total value of the public item pool to each each agent. This revelation is always truthful, and a player may only strategize over whether or not to reveal the signal.

In each of a finite number $R > 0$ of negotiation rounds, the players take turns, with player~1 moving first in each round. 
In its turn, a player $j$ can accept the latest offer from the other player (\textsc{deal}), end negotiations (\textsc{walk}), or make an offer-revelation combination of the form $(\offer,\cR)$. 
An offer denoted by
\begin{equation*}
    \offer \in \{ (\mathbf{p}_1, \mathbf{p}_2) \mid \mathbf{p}_1 + \mathbf{p}_2 = \mathbf{p} \}
\end{equation*}
is a proposed partition of the items where $\mathbf{p}_j$ is a vector of $\tau$ non-negative integers representing player~$j$'s allocated bundle by item type.
Revelation $\cR \in \{ \T, \F \}$ represents that player's decision to either disclose its signal (\T) or not (\F) to the other player during that turn. 
We also include a discount factor $\gamma \in (0,1]$ to capture preference for reaching deals sooner. 
Negotiation fails if a player chooses \textsc{walk} in any round $\rho \in\{1,\dotsc,R\}$ or $R$ rounds pass without any player choosing \textsc{deal}. 
In case of failure in round $\rho$, each player~$j$ receives a reward of $\gamma^{\rho} \mathbf{o}_j \cdot \mathbf{v}_j$ from its  outside offer. 
If a proposed partition $(\mathbf{p}_1, \mathbf{p}_2)$ is accepted in round $\rho$, then the reward to player $j$ is $\gamma^{\rho-1} \mathbf{p}_j \cdot \mathbf{v}_j$.

\begin{figure}[htbp]
\centering
\begin{subfigure}[t]{\textwidth}
\centering
\begin{tikzpicture}[thick,
    level 1/.style = {level distance = 10mm, sibling distance = 80mm},
    level 2/.style = {level distance = 10mm, sibling distance = 38mm},
    level 3/.style = {level distance = 10mm, sibling distance = 24mm},
    level 4/.style = {level distance = 15mm, sibling distance = 8mm},
    engine/.style = {inner sep = 1pt, above}]
    \node [draw, black, fill={chance}, regular polygon, regular polygon sides=3, rotate=180, inner sep=0.10cm] {}
    [black, ->]
    child { node [draw, black, fill={chance}, regular polygon, regular polygon sides=3, rotate=180, inner sep=0.10cm] (1) {} 
      child {node [draw, black, fill={chance}, regular polygon, regular polygon sides=3, rotate=180, inner sep=0.10cm] (2) {} 
        child {node [draw, black, fill={chance}, regular polygon, regular polygon sides=3, rotate=180, inner sep=0.10cm] (3) {} 
          child {node [draw, black, fill={player1}, circle] (4) {} 
          edge from parent node[engine, left] {$H$}}
          child {node [draw, black, fill={player1}, circle] (5) {} 
          edge from parent node[engine, right] {$L$}}
        edge from parent node[engine, sloped] {$v_2^1$}
          }
        child {node [draw, black, fill={chance}, regular polygon, regular polygon sides=3, rotate=180, inner sep=0.10cm] (6) {} 
          child {node [draw, black, fill={player1}, circle] (7) {} 
          edge from parent node[engine, left] {$H$}}
          child {node [draw, black, fill={player1}, circle] (8) {} 
          edge from parent node[engine, right] {$L$}}
          edge from parent node[engine, sloped] {$v_2^{\mid \cV \mid}$}
          }
          edge from parent node[engine, sloped] {$H$}
        }
      child {node [draw, black, fill={chance}, regular polygon, regular polygon sides=3, rotate=180, inner sep=0.10cm] (9) {} 
        child {node [draw, black, fill={chance}, regular polygon, regular polygon sides=3, rotate=180, inner sep=0.10cm] (10) {} 
          child {node [draw, black, fill={player1}, circle] (11) {} 
          edge from parent node[engine, left] {$H$}}
          child {node [draw, black, fill={player1}, circle] (12) {} 
          edge from parent node[engine, right] {$L$}}
        edge from parent node[engine, sloped] {$v_2^1$}}
        child {node [draw, black, fill={chance}, regular polygon, regular polygon sides=3, rotate=180, inner sep=0.10cm] (13) {} 
          child {node [draw, black, fill={player1}, circle] (14) {} 
          edge from parent node[engine, left] {$H$}}
          child {node [draw, black, fill={player1}, circle] (15) {} 
          edge from parent node[engine, right] {$L$}}
        edge from parent node[engine, sloped] {$v_2^{\mid \cV \mid}$}
        }
        edge from parent node[engine, sloped] {$L$}
        }
    edge from parent node[engine, sloped] {$v_1^1$}
    }
    child { node [draw, black, fill={chance}, regular polygon, regular polygon sides=3, rotate=180, inner sep=0.10cm] (16) {} 
      child {node [draw, black, fill={chance}, regular polygon, regular polygon sides=3, rotate=180, inner sep=0.10cm] (17) {} 
        child {node [draw, black, fill={chance}, regular polygon, regular polygon sides=3, rotate=180, inner sep=0.10cm] (18) {} 
        child {node [draw, black, fill={player1}, circle] (19) {} 
        edge from parent node[engine, left] {$H$}}
        child {node [draw, black, fill={player1}, circle] (20) {} 
        edge from parent node[engine, right] {$L$}}
        edge from parent node[engine, sloped] {$v_2^1$}
        }
        child {node [draw, black, fill={chance}, regular polygon, regular polygon sides=3, rotate=180, inner sep=0.10cm] (21) {} 
          child {node [draw, black, fill={player1}, circle] (22) {} 
          edge from parent node[engine, left] {$H$}}
          child {node [draw, black, fill={player1}, circle] (23) {} 
          edge from parent node[engine, right] {$L$}}
        edge from parent node[engine, sloped] {$v_2^{\mid \cV \mid}$}}
      edge from parent node[engine, sloped] {$H$}}
      child {node [draw, black, fill={chance}, regular polygon, regular polygon sides=3, rotate=180, inner sep=0.10cm] (24) {} 
      child {node [draw, black, fill={chance}, regular polygon, regular polygon sides=3, rotate=180, inner sep=0.10cm] (25) {} 
        child {node [draw, black, fill={player1}, circle] (26) {} 
        edge from parent node[engine, left] {$H$}}
        child {node [draw, black, fill={player1}, circle] (27) {} 
        edge from parent node[engine, right] {$L$}}
        edge from parent node[engine, sloped] {$v_2^1$}
        }
        child {node [draw, black, fill={chance}, regular polygon, regular polygon sides=3, rotate=180, inner sep=0.10cm] (28) {} 
          child {node [draw, black, fill={player1}, circle] (29) {} 
          edge from parent node[engine, left] {$H$}}
          child {node [draw, black, fill={player1}, circle] (30) {} 
          edge from parent node[engine, right] {$L$}}
        edge from parent node[engine, sloped] {$v_2^{\mid \cV \mid}$}}
      edge from parent node[engine, sloped] {$L$}}
      edge from parent node[engine, sloped] {$v_1^{\mid \cV \mid}$}
    };
    \path (1) -- node[auto=false]{{\LARGE \ldots}} (16);
    \path (3) -- node[auto=false]{{\normalsize \ldots}} (6);
    \path (10) -- node[auto=false]{\ldots} (13);
    \path (9) -- node[auto=false]{{\Large \ldots}} (17);
    \path (18) -- node[auto=false]{\ldots} (21); 
    \path (25) -- node[auto=false]{\ldots} (28);
    
\end{tikzpicture}
\end{subfigure}
\begin{subfigure}[t]{\textwidth}
\centering
\vspace{1em}
    \fbox{\begin{tabular}{rr}
\small{\textcolor{chance}{\TriangleDown}} & \footnotesize Chance Nodes \\
\small{\textcolor{player1}{\CircleSolid}} & \footnotesize Player 1 Nodes \\
\end{tabular}}
\end{subfigure}
    \caption{EFG representation of \dondG, from the start of the game until just before negotiations start with player~1's initial offer to player~2.}
    \label{fig:full_bargain_chap3}
\end{figure}

\begin{figure}[htp!]
\centering
\begin{subfigure}[t]{\textwidth}
\centering
\begin{tikzpicture}[thick,
    level 1/.style = {level distance = 10mm, sibling distance = 70mm},
    level 2/.style = {level distance = 15mm, sibling distance = 22mm},
    level 3/.style = {level distance = 28mm, sibling distance = 16mm},
    engine/.style = {inner sep = 1pt, above}]
    \node [draw, black, fill={chance}, regular polygon, regular polygon sides=3, rotate=180, inner sep=0.10cm] {} 
    [black, ->]
      child { node [draw, black, fill={player1}, circle] (1) {}
        child { node [draw, black, fill={player2}, diamond] (2) {}
          child { node [draw, black, fill={player1}, circle] (3) {}
          edge from parent node[engine, left] {{\small $a_1$}}}
          child { node [draw, black, fill={player1}, circle] (4) {}
          edge from parent node[engine, left] {{\small $a_m$}}}
          child { node [draw, black, fill={terminal}, regular polygon, regular polygon sides=3, inner sep=0.08cm, left=12em, above=3.5em] (5) {}
          edge from parent node[engine, sloped] {{\footnotesize $\textsc{walk}$}}}
          child { node [draw, black, fill={terminal}, regular polygon, regular polygon sides=3, inner sep=0.08cm, left=12em, above=4.5em] (6) {}
          edge from parent node[engine, sloped] {{\footnotesize $\textsc{deal}$}}}
        edge from parent node[engine, left] {{\small $a_1$}}} 
        child { node [draw, black, fill={player2}, diamond, right=2.5em] (7) {}
          child { node [draw, black, fill={player1}, circle, right=2em] (8) {}
          edge from parent node[engine, left] {{\small $a_1$}}}
          child { node [draw, black, fill={player1}, circle, right=2em] (9) {}
          edge from parent node[engine, left] {{\small $a_m$}}}
          child { node [draw, black, fill={terminal}, regular polygon, regular polygon sides=3, inner sep=0.08cm, above=3.5em] (10) {}
          edge from parent node[engine, sloped] {{\footnotesize $\textsc{walk}$}}}
          child { node [draw, black, fill={terminal}, regular polygon, regular polygon sides=3, inner sep=0.08cm, left=4em, above=4.5em] (11) {}
          edge from parent node[engine, sloped] {{\footnotesize $\textsc{deal}$}}}
        edge from parent node[engine, left] {{\small $a_m$}}} 
      edge from parent node[engine, sloped] {$H$}} 
      child { node [draw, black, fill={player1}, circle] (14) {} 
        child { node [draw, black, fill={player2}, diamond] (15) {}
             child { node [draw, black, fill={player1}, circle, right=2em] (16) {}
          edge from parent node[engine, left] {{\small $a_1$}}}
          child { node [draw, black, fill={player1}, circle, right=2em] (17) {}
          edge from parent node[engine, left] {{\small $a_m$}}}
          child { node [draw, black, fill={terminal}, regular polygon, regular polygon sides=3, inner sep=0.08cm, above=3.5em] (18) {}
          edge from parent node[engine, sloped] {{\footnotesize $\textsc{walk}$}}}
          child { node [draw, black, fill={terminal}, regular polygon, regular polygon sides=3, inner sep=0.08cm, left=2em, above=4.5em] (19) {}
          edge from parent node[engine, sloped] {{\footnotesize $\textsc{deal}$}}}
        edge from parent node[engine, left] {{\small $a_1$}}} 
        child { node [draw, black, fill={player2}, diamond, right=2.5em] (20) {}
          child { node [draw, black, fill={player1}, circle, right=3em] (21) {}
          edge from parent node[engine, left] {{\small $a_1$}}}
          child { node [draw, black, fill={player1}, circle, right=3em] (22) {}
          edge from parent node[engine, left] {{\small $a_m$}}}
          child { node [draw, black, fill={terminal}, regular polygon, regular polygon sides=3, inner sep=0.08cm, above=3.5em] (23) {}
          edge from parent node[engine, sloped] {{\footnotesize $\textsc{walk}$}}}
          child { node [draw, black, fill={terminal}, regular polygon, regular polygon sides=3, inner sep=0.08cm, right=1em, above=4.5em] (24) {}
          edge from parent node[engine, sloped] {{\footnotesize $\textsc{deal}$}}}
        edge from parent node[engine, left] {{\small $a_m$}}} 
      edge from parent node[engine, sloped] {$L$}}
;
    \path (2) -- node[auto=false]{{\large \ldots}} (7);
    \path (3) -- node[auto=false]{{\large \ldots}} (4);
    \path (8) -- node[auto=false]{{\large \ldots}} (9);
    \path (15) -- node[auto=false]{{\large \ldots}} (20);
    \path (16) -- node[auto=false]{{\large \ldots}} (17);
    \path (21) -- node[auto=false]{{\large \ldots}} (22);
    
\end{tikzpicture}
\end{subfigure}~\\
\begin{subfigure}[t]{\textwidth}
\centering
\vspace{0.3em}
\fbox{\begin{tabular}{rr}
\small{\textcolor{chance}{\TriangleDown}} & \footnotesize Chance Nodes \\
\small{\textcolor{player1}{\CircleSolid}} & \footnotesize Player 1 Nodes \\
\small{\textcolor{player2}{\DiamondSolid}} & \footnotesize Player 2 Nodes \\
\small{\textcolor{terminal}{\TriangleUp}} & \footnotesize Terminal Nodes
\end{tabular}}
\end{subfigure}
\caption{The first round of negotiation in \dondG.}
\label{fig:full_bargain2_chap3}
\end{figure}

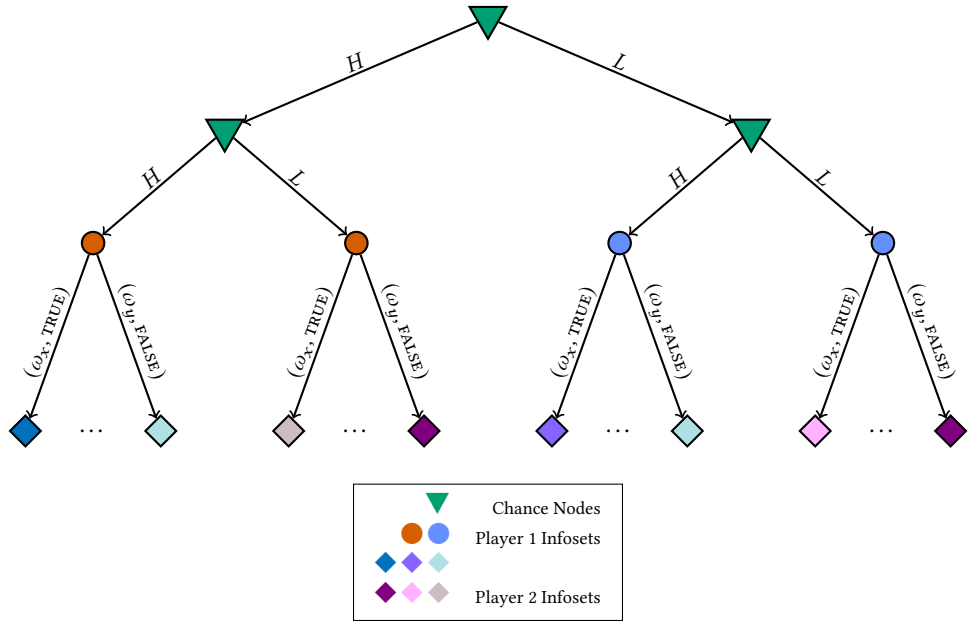
\begin{figure}[ht!]
    \centering
\begin{subfigure}[t]{\textwidth}
\centering
\begin{tikzpicture}[thick,
    level 1/.style = {level distance = 15mm, sibling distance = 70mm},
    level 2/.style = {level distance = 15mm, sibling distance = 35mm},
    level 3/.style = {level distance = 25mm, sibling distance = 18mm},
    engine/.style = {inner sep = 1pt, above}]
    \node [draw, black, fill={chance}, regular polygon, regular polygon sides=3, rotate=180, inner sep=0.10cm] {}
    [black, ->]

    child {node [draw, black, fill={chance}, regular polygon, regular polygon sides=3, rotate=180, inner sep=0.10cm] (1) {} 
      child {node [draw, black, fill={player1}, circle] (2) {} 
        child {node [draw, black, fill={player2}, diamond] (3) {}
        edge from parent node[engine, sloped] {{\small $(\offer_x, \T)$}}}
        child {node [draw, black, fill={player2_2}, diamond] (4) {}
        edge from parent node[engine, sloped] {{\small $(\offer_y, \F)$}}}
      edge from parent node[engine, sloped] {$H$}}
      child {node [draw, black, fill={player1}, circle] (6) {} 
        child {node [draw, black, fill={player2_5}, diamond] (7) {}
        edge from parent node[engine, sloped] {{\small $(\offer_x, \T)$}}}
        child {node [draw, black, fill={player2_3}, diamond] (8) {}
        edge from parent node[engine, sloped] {{\small $(\offer_y, \F)$}}}
      edge from parent node[engine, sloped] {$L$}}
      edge from parent node[engine, sloped] {$H$}}
    child {node [draw, black, fill={chance}, regular polygon, regular polygon sides=3, rotate=180, inner sep=0.10cm] (10) {}
    child {node [draw, black, fill={player1_1}, circle] (11) {} 
        child {node [draw, black, fill={player2_1}, diamond] (12) {}
        edge from parent node[engine, sloped] {{\small $(\offer_x, \T)$}}}
        child {node [draw, black, fill={player2_2}, diamond] (13) {}
        edge from parent node[engine, sloped] {{\small $(\offer_y, \F)$}}}
      edge from parent node[engine, sloped] {$H$}}
      child {node [draw, black, fill={player1_1}, circle] (15) {} 
        child {node [draw, black, fill={player2_4}, diamond] (16) {}
        edge from parent node[engine, sloped] {{\small $(\offer_x, \T)$}}}
        child {node [draw, black, fill={player2_3}, diamond] (17) {}
        edge from parent node[engine, sloped] {{\small $(\offer_y, \F)$}}}
      edge from parent node[engine, sloped] {$L$}}
      edge from parent node[engine, sloped] {$L$}
    };

   \path (3) -- node[auto=false]{{\large \ldots}} (4);
    \path (7) -- node[auto=false]{{\large \ldots}} (8);
    \path (12) -- node[auto=false]{{\large \ldots}} (13);
    \path (16) -- node[auto=false]{{\large \ldots}} (17);
\end{tikzpicture}

\end{subfigure}
\begin{subfigure}[t]{\textwidth}
\centering
\vspace{0.3em}
\fbox{\begin{tabular}{rr}
\small{\textcolor{chance}{\TriangleDown}} & \footnotesize Chance Nodes \\
\small{\textcolor{player1}{\CircleSolid} \textcolor{player1_1}{\CircleSolid}} & \footnotesize Player 1 Infosets \\
\small{\textcolor{player2}{\DiamondSolid} \textcolor{player2_1}{\DiamondSolid} \textcolor{player2_2}{\DiamondSolid}} \\ \small{\textcolor{player2_3}{\DiamondSolid} \textcolor{player2_4}{\DiamondSolid} \textcolor{player2_5}{\DiamondSolid}} & \footnotesize Player 2 Infosets
\end{tabular}}
\end{subfigure}
\caption{Illustration of the effect of player~1's revelation-decision $\cR$ on player~2's infoset structure at start of \dondG.}
\label{fig:DOND_no_offer_reveal_chap3}
\end{figure}

We visualize \dondG\ as an EFG with the help of Figures~\ref{fig:full_bargain_chap3} through \ref{fig:DOND_no_offer_reveal_chap3}. 
Figure~\ref{fig:full_bargain_chap3} shows the partial game tree starting at the root and ending just before negotiation begins, and Figure~\ref{fig:full_bargain2_chap3} shows the subtree beginning at one of the chance nodes following player~2's sampled valuation from Figure~\ref{fig:full_bargain_chap3}. Each subsequent round of negotiation repeats starting from one of player~1's nodes, for a total of $R$ rounds.

Figure~\ref{fig:DOND_no_offer_reveal_chap3} offers a alternative view of the game where the sampling of the valuation vectors is suppressed and the formation of player~$2$'s information sets based on player~$1$'s actions is emphasized. First, the players' outside offers are sequentially sample by Nature. Next, player~1 chooses an action, comprising an offer $\offer$ and revelation $\cR$. Player~2 has four distinguishable histories that result when player~1 takes action $(\omega_x, \T)$, choosing to reveal its signal to player~2: $(H, H, (\omega_x, \T))$, $(H, L, (\omega_x, \T))$, $(L, H, (\omega_x, \T))$, $(L, L, (\omega_x, \T))$. When player~1 takes action $(\omega_y, \F)$, choosing not to reveal its signal to player~2, two non-singleton information sets are induced for player~2: one containing the histories $(H, H, (\omega_y, \F))$ and $(L, H, (\omega_y, \F))$, and the other containing the histories $(H, L, (\omega_y, \F))$ and $(L, L, (\omega_y, \F))$.

For our experiments on \dondG (Section~\ref{sec:pbe_exp}), we set $\tau=3$, $\bar{V}=10$, $\thresh=5$, $\gamma=0.99$, and $R=5$.
We generated five unique sets of the remaining parameters: $\mathbf{p} \in \{(2,0,3),(3,1,2),(1,2,2),(1,4,2),(0,0,5)\}$, and $(\mathbf{v}_1, \mathbf{v}_2)$, $P_1$, $P_2$ sampled uniformly at random from their respective supports. 

\section{All plots for experiments in Section~\ref{sec:tepsro_expts}}\label{sec:app_pbe_mss}
In this section, we provide the complete set of regret curves versus TE-PSRO epochs obtained from our experiments in Section~\ref{sec:tepsro_expts} comparing (unrefined) NE and PBE as the chosen MSS, for \gengoofK4 and \dondG. We used NE for the EVAL solution $\bsigma^*$. In addition to the details provided in Section~\ref{sec:tepsro_expts}, note that we varied $M$ in $\{1, 2, 4, 8, 16\}$ for \dondG and $\{1, 2, 4, 8\}$ for \gengoofK4.


\begin{figure}[H]
    \centering
    \begin{subfigure}[b]{0.49\textwidth}
     \includegraphics[width=0.85\textwidth]{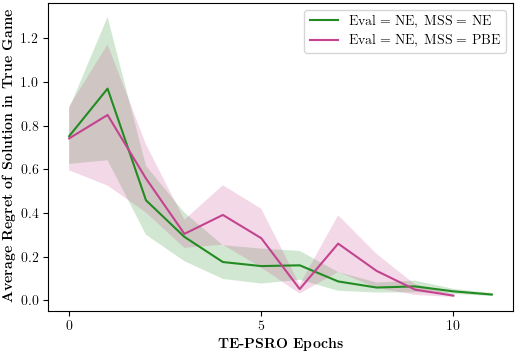}
    \caption{$M = 1$ \label{fig:app_pbe_barg_M1}}
    \end{subfigure}~
    \begin{subfigure}[b]{0.49\textwidth}
     \includegraphics[width=0.85\textwidth]{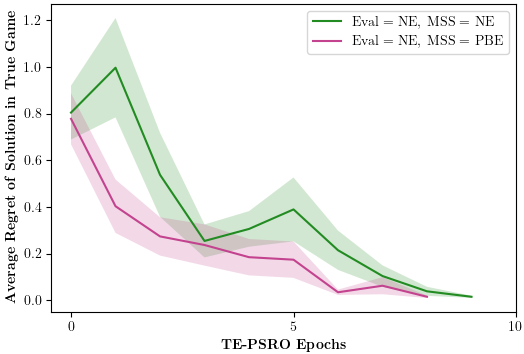}
    \caption{$M = 2$ \label{fig:app_pbe_barg_M2}}
    \end{subfigure}\\
    \begin{subfigure}[b]{0.49\textwidth}
    \includegraphics[width=0.85\textwidth]{chapters/plots_chap8/plots_nash_refine_pbe/bargaining_game/true_regret_9_3_pbe_M4_no_NF.png}
    \caption{$M = 4$ \label{fig:app_pbe_barg_M4}}
    \end{subfigure}~
    \begin{subfigure}[b]{0.49\textwidth}
    \includegraphics[width=0.85\textwidth]{chapters/plots_chap8/plots_nash_refine_pbe/bargaining_game/true_regret_9_3_pbe_M8_no_NF.png}
    \caption{$M = 8$ \label{fig:app_pbe_barg_M8}}
    \end{subfigure}\\
    \begin{subfigure}[b]{0.49\textwidth}
    \includegraphics[width=0.85\textwidth]{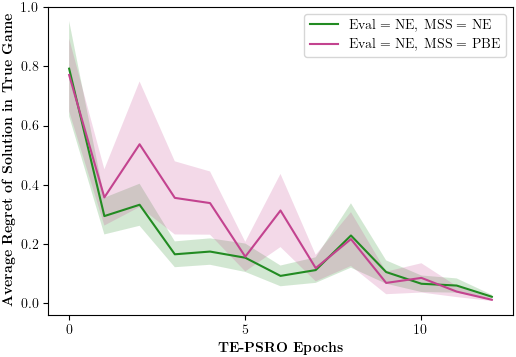}
    \caption{$M = 16$ \label{fig:app_pbe_barg_M16}}
    \end{subfigure}
    \caption{Average regret of $\bsigma^*$ evaluated in \dondG \hspace{0.1em} over the course of TE-PSRO's runtime, using NE or PBE as the MSS.}
    \label{fig:app_barg_pbe_mss}
\end{figure}


\begin{figure}[H]
    \centering
    \begin{subfigure}[b]{0.49\textwidth}
     \includegraphics[width=0.85\textwidth]{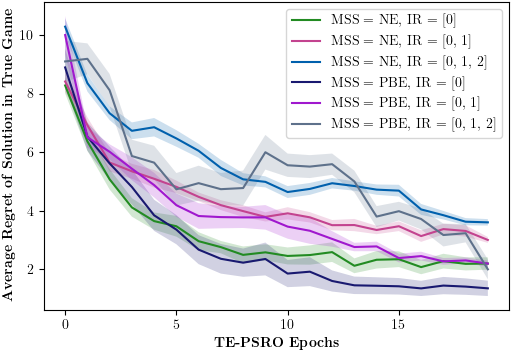}
     \caption{$M = 1$ \label{fig:app_pbe_abs4_M1}}
    \end{subfigure}~
    \begin{subfigure}[b]{0.49\textwidth}
     \includegraphics[width=0.85\textwidth]{chapters/plots_chap8/plots_nash_refine_pbe/abstract_4rounds/true_regret_4rounds_ne_eval_pbe_mss_M2.png}
     \caption{$M = 2$ \label{fig:app_pbe_abs4_M2}}
    \end{subfigure}\\
    \begin{subfigure}[b]{0.49\textwidth}
     \includegraphics[width=0.85\textwidth]{chapters/plots_chap8/plots_nash_refine_pbe/abstract_4rounds/true_regret_4rounds_ne_eval_pbe_mss_M4.png}
     \caption{$M = 4$ \label{fig:app_pbe_abs4_M4}}
    \end{subfigure}~
    \begin{subfigure}[b]{0.49\textwidth}
     \includegraphics[width=0.85\textwidth]{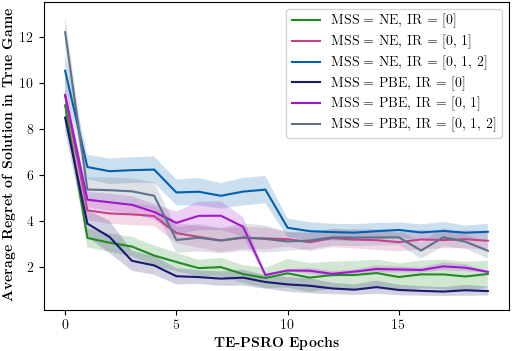}
     \caption{$M = 8$ \label{fig:app_pbe_abs4_M8}}
    \end{subfigure}\\
    \caption{Average regret of $\bsigma^*$ evaluated in \gengoofK4 \hspace{0.1em} over the course of TE-PSRO's runtime, using NE or PBE as the MSS.}
    \label{fig:app_abs4_pbe_mss}
\end{figure}

\section{Review of game-theoretic solution concepts with belief}\label{app:refinements}

\subsection{Consistency and Sequential Equilibria}\label{sec:seq_eq}

A couple of decades after \citet{selten65} introduced SPE, \citet{kw82} introduced a new theoretical solution concept called \term{sequential equilibrium} (SE). Sequential equilibrium, as well as the other solution concepts outlined in this chapter, is based on extending the idea of a NE being induced at every decision point in an EFG that is required by SPE and ensured via backward induction to arbitrary games of imperfect information \citep{kohlberg86}.

An assessment $(\bsigma, \mu)$ must satisfy three conditions in order to be considered a SE. First, $(\bsigma, \mu)$ must satisfy sequential rationality (Definition~\ref{def:seq_rat}). Second, every distribution of $\mu$ must be satisfy Bayes' rule given $\bsigma$ and $X$ (since information sets can be induced by unknown chance outcomes also). Specifically, at every information set $I \in \cI_j$ with reach probability $r(I, \bsigma) > 0$, for all players $j \in N \setminus \{ 0 \}$,
\begin{equation*}
    \mu(h) = \frac{r(h, \bsigma)}{r(I, \bsigma)} = \frac{r(h, \bsigma)}{\sum_{h' \in I} r(h', \bsigma)}
\end{equation*}
for all $h \in I$.

The third and final SE condition that an assessment $(\bsigma, \mu)$ must satisfy is referred to by \citet{kw82} as \term{consistency}; to distinguish from other concepts of consistency, I refer to their definition as as Kreps-Wilson (KW) consistency. Consistency ensures that the beliefs are sensible in all parts of the game, even for the information sets that are not reached given the assessment. An assessment is \textit{KW-consistent} if given an infinite sequence of completely mixed strategy profiles $\left[ \bsigma^1, \ldots, \bsigma^m \right]$, and an associated belief system $\mu^m$ for each profile defined according to Bayes' rule,
\begin{equation*}
    \lim_{m \rightarrow \infty} \left( \bsigma^m, \mu^m \right) = (\bsigma, \mu).
\end{equation*}

The sequence in question is a theoretical infinite sequence of assessments drawn from the set of all possible assessments, not a sequence resulting from a series of iteratively tweaked assessments returned by RL, a no-regret learning algorithm, or some other game-theoretic algorithm being applied to the game in question. The KW-consistent assessment $(\bsigma, \mu)$ is the closure of the set of all possible assessments. A solution criterion that is based on topological spaces is not very straightforward for direct computation or for verification of an assessment candidate. Later works have attempted to give a different solution concept that allowed for the incorporation of beliefs via Bayes' rule while enforcing the sensibility of the paths taken through the game tree given $\bsigma$ and $\mu$, albeit more simply.

\subsection{Weak Sequential Equilibrium}\label{sec:weak_se}

\citet{myerson91} gave a new solution concept called \term{weak sequential equilibrium} that did not include the consistency requirement as proposed by \citet{kw82}. In order to be a weak sequential equilibrium, a solution candidate $(\bm{\sigma}, \mu)$ must satisfy two conditions. First, $(\bm{\sigma}, \mu)$ must induce a NE at every information set \textit{reached with positive probability}, for every player. Second, $\mu$ must be defined for nodes on the paths defined by $(\bm{\sigma}$ through the game tree according to Bayes’ rule. However, since no restrictions are imposed by weak sequential equilibrium on the assessment at information sets that are off the equilibrium path, this solution concept does not actually guarantee that sequential rationality is satisfied.

Consider as an example a candidate assessment for the game depicted in Figure~\ref{fig:not_agm_consist} with the following strategy profile:
\begin{align*}
    \sigma_1(I(\emptyset)) &= c \\
    \sigma_2(I(b)) &= d \\
    \sigma_3(I(c)) &= f,\ \sigma_3(I(\mathit{bd})) = h.
\end{align*}

Weak SE would impose no requirements upon Player 3's belief at the highlighted information set since it is off the equilibrium path (highlighted in green). This means that $\mu(\mathit{bd})$ and $\mu(\mathit{be})$ do not have to be equal to zero, for instance, or even add up to 1. Weak SE also would not restrict the choices of Player 2 and Player 3 in the other subgame that is rooted at the node with history $B$. It is clear that although this solution concept is more lenient than SE, this particular solution candidate could satisfy both conditions for weak SE, but fail to satisfy the conditions required for the assessment to be optimal at every part of the game, including information sets that are off the equilibrium path.

\subsection{Perfect Bayesian Equilibrium for Multi-Stage Signaling Games}\label{sec:pbe_msg}

\citet{pbe91} introduced a more restrictive solution concept for dynamic games of imperfect information that is close to sequential equilibrium. Their version does not include the condition of consistency that is imposed by sequential equilibrium upon paths that are taken through the game tree with zero probability. However, this solution concept was developed specifically for a certain class of games known by several names: \textit{multi-stage signaling games}, \textit{multi-period games with observed actions}, \textit{multi-period games of imperfect information}, \textit{multi-period games with independent types} \citep{pbe91, game_theory_ft, fud_levine_83}. This class was chosen for the games' ability to showcase the complications introduced by imperfect information. The new solution concept introduced via this class of games was named the perfect Bayesian equilibrium by \citeauthor{pbe91}; we will call it PBE-MSG to avoid confusion with the solution concept of interest in our paper.

At the very beginning of the multi-stage signaling game described by \citet{pbe91,game_theory_ft}, each player $j$ is first assigned a private type $\theta_j$ from a finite set $\Theta_j$. The player types are all independent of each other. Each player learns his own type but learns nothing about the other players' types. Then, the game proceeds in a series of stages, each of which is consistent with the format of a simple signaling game. Given their own type and nothing else at stage $t$, each player simultaneously chooses a strategy $\sigma^t_j$ that maps the history $h^t$ up to stage $t$ and $\theta_j$ to its action space $A_j(h^t)$. A history $h^t$ includes each $n$-long vector of all player actions $a^t = \left[ a^t_j \right]_{j = 1}^n$ chosen from the beginning of the game up to stage $t$. Finally, all actions are revealed at the end of the stage, and payoffs are made given the new history $h^{t + 1} = \left[ a^0, \ldots, a^t \right]$ and the vector of player types. The total number of stage repetitions is given by the parameter $T$.

Each player $j$ formulates an independent posterior belief $\mu$ about the vector of other players' types $\theta_{-j}$. It is assumed that the posterior beliefs are also independent and that all types of player $j$ have the same beliefs, meaning even unexpected observations to not make any of the players think that his opponents' types are correlated: 
\begin{equation*}
    \mu_j \left( \theta_{-j} \given \theta_j, h^t \right) = \prod_{k \neq j} \mu\left( \theta_k \given h^t \right).
\end{equation*}
$\mu$ is updated for each player $j$ at the end of each stage according to Bayes' rule, for each $\theta_k \in \theta_{-j}$
\begin{equation*}
    \mu_j \left( \theta_k \given (h^t, a^t) \right) = \frac{\mu_j (\theta_k \given h^t) \cdot \sigma_k \left( a^t_k \given h^t, \theta_k \right)}{\sum_{\theta'_k \in \theta_{-j}} \mu_j \left( \theta'_k \given h^t \right) \cdot \sigma_k \left( a^t_k \given h^t, \theta'_k \right)}.
\end{equation*}

The beliefs are updated even for histories that have probability 0 at stage $t$, which is slightly stronger than simply requiring that Bayes' rule be applied consistently for history-action combinations with positive probability. $\mu$ should naturally also be consistent with the player types all being independent. Another requirement for updating beliefs is that even unexpected observations of opponent actions do not induce $\mu$ to reflect correlation between opponent player types. 
PBE-MSG also requires that the expected payoff from $(\bsigma, \mu)$ induces a NE for each stage $t$, thus ensuring that $(\bsigma, \mu)$ is a subgame-perfect equilibrium for this setting. Finally, a player $j$'s belief about player $k$'s type should not change because the actions of players other than $j$ or $k$ deviated. Altogether, these requirements provided a promising format for the general concept of perfect Bayesian equilibrium for imperfect information games (Section~\ref{sec:pbe}).

\subsection{Other Related Work}\label{sec:related_chap7}

An alternative NE refinement for imperfect-information EFGs was proposed by \citet{battigalli96} that was also called PBE. But, this refinement is in terms of the even more complex solution space of tree-extended assessments $(\nu,\bsigma,\mu)$, where $\nu$ is a conditional probability system on the set of terminal nodes.

The PBE concept that we use is a weaker version of the sequential equilibrium (SE), with a more lenient and verifiable definition of consistency. \citet{azhar05} gave an exponential-time algorithm that outlined all possible bases of a given EFG and $(\bsigma, \mu)$ that satisfy consistency and gave a system of polynomial equations and inequalities that specified the sequential equilibria of each basis. A basis was defined as a set of terminal nodes reached with nonzero probability according to $(\bsigma, \mu)$, of which there could be exponentially many. \citet{turocy2010computing}, extending the work of \citet{mckelvey1998quantal}, provided a method to approximate a sequential equilibrium arbitrarily well via a sequence of agent quantal response equilibria. 
\citet{panozzo14} proposed multiple approaches for verifying the sequential equilibria of an imperfect-information EFG algorithmically. To the best of our knowledge, although these algorithms demonstrate the feasibility of computing an (approximate) SE of of an imperfect-information game, their scalability has not been adequately established in the literature.

More recently, \citet{Thoma25} introduced a best-response based algorithm for computing $\varepsilon$-perfect Bayesian equilibria in sequential auctions with incomplete information and combinatorial bidding; since auctions are a form of multi-stage game, this work demonstrated how to algorithmically compute an approximate PBE-MSG. 
\citet{graf_sequent_eq24} introduced an algorithm that symbolically solved imperfect-information games for sequential equilibria using a finite system of polynomial equations and inequalities; however, the algorithm is viable only for small games and not scalable.

Beliefs when incorporated for games of imperfect information are always probabilistic rather than qualitative; by contrast, a qualitative player belief system might utilize the logic of Kripke frames regarding the consequences of taking alternative actions in a perfect information game in order to assess rationality \citep{bonanno_beliefs13}. Furthermore, given that our considerations are restricted to non-cooperative extensive-form games, player strategies are by definition causally independent, even if they may be epistemically correlated with past actions \citep{bernheim84}. It follows that the beliefs must also be ``uncorrelated probabilistic assessments of their opponents' choices" \citep{brandenburger87}.

The notion of consistency as a necessary part of the definition of sequential equilibrium is justified by many other works \cite{kohlberg_consist97,swinkels_consist93,battigalli96}, even if the topographical definition given by \citet{kw82} is not algorithmically verifiable. Intuitively, $\bsigma$ can be understood as a roadmap of paths to be taken by the players in the future starting from any decision point while $\mu$ can be understood as a roadmap of paths taken in hindsight. Consistency is required in order to ensure that the possible paths specified by both do not contradict each other, particularly for information sets that happen to be off the equilibrium path.

Several works following \citet{kw82} have introduced other approaches to the consistency requirement. A weaker measure of consistency known as preconsistency was introduced by \citet{hendon94}, requiring only that consistency be maintained between information sets belonging to the same player. \citet{perea02} introduced an even weaker restriction known as updating consistency. Both works demonstrated that their respective formulations of consistency were sufficient for a given assessment to be sequentially rational, given that the assessment also satisfied a game-theoretic principle known as the one-shot deviation principle (defined in Section~\ref{sec:pbe_proof}). \citet{kohlberg_consist97} demonstrated that the set of consistent assessments for any given game could be described through a finite system of polynomial inequalities; it is important to note that this work treated beliefs as being part of the assessment of a game, derived by an outside observer based on the strategies employed by the players. 
\citet{pimienta_consist14} characterized the set of consistent assessments given an EFG, with a particular focus on how to guarantee that an assessment whose beliefs satisfy Bayes' rule on the equilibrium path is also consistent at any information sets off the equilibrium path.


\end{document}